%% file: journal_IT_submit.tex
\documentclass[12pt, draftclsnofoot, onecolumn,romanappendices]{IEEEtran}

\usepackage{bm,cite}
\usepackage{amsmath}
\usepackage{amsthm}
\usepackage{amsbsy}
\usepackage{epsfig}
\usepackage{latexsym}
\usepackage{amssymb}
\usepackage{wasysym}
\usepackage{placeins}
\usepackage{color}
\usepackage{rotating}
\usepackage[usenames,dvipsnames]{xcolor}
\usepackage{dsfont}
\usepackage{caption}
\usepackage{subcaption}
\usepackage{algorithm}
\usepackage{algorithmic}
\usepackage{bbold}

\usepackage{tikz}
\usepackage{multirow}
\usepackage{slashbox}
\usepackage{diagbox}

\usetikzlibrary{shapes,arrows,snakes}
\usetikzlibrary{calc}

\IEEEoverridecommandlockouts
\input{Definitions}

\begin{document}
\title{Robust Regularized ZF in Cooperative Broadcast Channel under Distributed CSIT}
\author{\IEEEauthorblockN{Qianrui Li\IEEEauthorrefmark{3}, Paul de Kerret\IEEEauthorrefmark{1}, David Gesbert\IEEEauthorrefmark{1}, and Nicolas Gresset\IEEEauthorrefmark{3}}

\IEEEauthorblockA{\IEEEauthorrefmark{3}Mitsubishi Electric R\&D Centre Europe\\}
\IEEEauthorblockA{\IEEEauthorrefmark{1}
Communication Systems Department, EURECOM\\}

\thanks{Part of this work has been presented at the 2015 IEEE International Symposium on Information Theory (ISIT 2015) and $53$rd Annual Allerton Conference on Communication, Control and Computing (Allerton 2015). David Gesbert and Paul de Kerret are supported by the ERC under the European Union’s Horizon 2020 research and innovation program (Agreement no. 670896). }
}

\maketitle
\begin{abstract}
In this work, we consider the sum rate performance of joint processing coordinated multi-point transmission network (JP-CoMP, a.k.a Network MIMO) in a so-called distributed channel state information (D-CSI) setting. In the D-CSI setting, the various transmitters (TXs) acquire a local, TX-dependent, estimate of the global multi-user channel state matrix obtained via terminal feedback and limited backhauling. The CSI noise across TXs can be independent or correlated, so as to reflect the degree to which TXs can exchange information over the backhaul, hence allowing to model a range of situations bridging fully distributed and fully centralized CSI settings. In this context we aim to study the price of CSI distributiveness in terms of sum rate at finite SNR when compared with conventional centralized scenarios. We consider the family of JP-CoMP precoders known as regularized zero-forcing (RZF). We conduct our study in the large scale antenna regime, as it is currently envisioned to be used in real 5G deployments. It is then possible to obtain accurate approximations on so-called deterministic equivalents of the signal to interference and noise ratios. Guided by the obtained deterministic equivalents, we propose an approach to derive a RZF scheme that is robust to the distributed aspect of the CSI, whereby the key idea lies in the optimization of a TX-dependent power level and regularization factor.  Our analysis confirms the improved robustness of the proposed scheme with respect to CSI inconsistency at different TXs, even with moderate number of antennas and receivers (RXs).
\end{abstract}


\section{Introduction}

Joint processing CoMP, whereby multiple cooperating TXs share the data streams and perform joint precoding \cite{Gesbert2010}, are considered for use in current and next generation wireless networks. Theoretically, with perfect data and CSI sharing, TXs at different locations can be seen as a unique virtual multiple-antenna array serving all RXs in a multiple-antenna broadcast channel (BC) fashion and well known precoding algorithms from the literature can be used \cite{Karakayali2006}. However, in real systems both the feedback through the wireless medium and the information exchange through the backhaul place a burden on overall resources and must be limited.

Joint processing CoMP under limited feedback and imperfect backhaul (or fronthaul for cloud radio access network, a.k.a C-RAN systems) has been investigated in many works. In \cite{Sanderovich2009,Simeone2009}, the capacity limited backhaul is considered and an information theoretic analysis of the system performance for joint processing CoMP is provided. In \cite{zhou2016optimal,park2016multihop,quek2017cloud}, the compress-and-forward schemes, cooperative beamforming and resource allocation for a C-RAN with capacity-limited fronthaul links are considered. In \cite{Wagner2012,Jindal2006,MaddahAli2012}, the effect of imperfect CSIT due to limited feedback and/or delay is investigated in a single TX multiple antennas broadcast channel setting. In \cite{Marsch2008,zhao2013}, precoder designs for the joint processing CoMP with limited backhaul are provided. However, most of these contributions typically assume a \emph{centralized} CSIT setting, i.e., the precoding is done on the basis of a \emph{single} imperfect channel estimate which is commonly known at every TX. 

This assumption of a centralized computing unit is relevant in the so-called C-RAN architecture, yet it is more and more challenged in other forms of networks where a pre-existing optical fiber backhaul is lacking or is considered too expensive in terms of CAPEX. Other emerging deployment scenarios are those with a fully heterogeneous infrastructure where the network's edge is composed of not just fixed macro base stations but also small cell base stations, mobile (possibly flying \cite{mozaffari2015drone}) access points or relays.
In such settings, exchanging CSI over limited and unreliable backhaul is likely to lead to additional quantization noise and latencies. As a result,  the global downlink CSI estimate collected by any TX is unique to that TX, although the CSI noise can exhibit some degree of correlation from TX to TX. In the rest of this paper, we refer to this setting as a Distributed CSI setting, which considers implicitly the possible correlation between the estimates.
In this context we are interested in the design of a distributed precoder whereby each
TX computes the elements of the precoder used for its transmission based solely on its own channel estimate.

From an information theoretic perspective, the study of joint processing CoMP in D-CSI setting raises several intriguing and challenging questions.

First, while the JP-CoMP with perfect user message sharing is akin to the information theoretic MISO broadcast channel, the capacity region of the broadcast channel under a general D-CSI setting is unknown. In \cite{dekerret2012_TIT}, a rate characterization at high SNR is carried out using DoF analysis for the two TXs scenario. This study highlights the severe penalty caused by the lack of a consistent CSI shared by the cooperating TXs from a DoF point of view, when using a conventional precoder. It was also shown that classical RZF \cite{Peel2005} do not restore the DoF. Although a new DoF-restoring decentralized precoding strategy was presented in \cite{dekerret2012_TIT} for the two TXs case, only partial results are known for the case of an arbitrary number of users\cite{dekerret2016_ISIT}. Furthermore, at finite SNR, the problem of designing precoders that optimally tackle the D-CSI setting is fully open. The use of conventional linear precoders that are unaware of the D-CSI structure is expected to yield a significant loss with respect to a centralized (and imperfect) CSI setting. Hence, an important question is how to reduce the losses due to the D-CSI configuration, i.e., how to derive a D-CSI-robust precoding scheme.

In this work, we study the average rate achieved when the number of transmit antennas and the number of receive antennas jointly grow large with a fixed ratio, thus allowing to use efficient tools from the field of random matrix theory (RMT). Although RMT has been applied in many works to the analysis of wireless communications [See~\cite{Hochwald2002,Tulino2007,Wagner2012,Muller2013,Couillet2011} among others], its role in helping to analyze cooperative systems with distributed information has received little attention before.

In this work, our main contribution are threefold:
\begin{itemize}
\item A novel general D-CSI channel model that allows to study distributed CoMP networks ranging from fully distributed to fully centralized is introduced. 
\item A deterministic equivalent of the SINR in D-CSI setting in the limit of a large number of antennas is derived.
\item Building upon this deterministic equivalent, the sum rate maximization regularization coefficient for the RZF precoder and the local optimal power allocation for each TX under a total power constraint can be found. This leads to a robust distributed RZF precoder design for the D-CSI setting. The regularization coefficient can either be optimized individually by each TX or be found by a low complexity heuristic algorithm assuming that a single common regularization coefficient is used at all TX. Simulations show that the low complexity approach approximates well the performance of the per-TX individually optimization.    
\end{itemize}

\textit{Notations:} In the following, boldface lower-case and upper-case characters denote vectors and matrices, respectively. The operator $(.)^{\Te},(.)^{\He},\trace(.),\E(.)$ denote transpose, conjugate transpose, trace and expectation, respectively. The $N\times N$ identity matrix is denoted $\bI_N$. The notation $\LSB\bA\RSB_{i,j}, \LSB\bb\RSB_{i}$ denotes the $(i,j)$th entry of matrix $\bA$  and the $i$th entry of vector $\bb$, respectively. $\diag(.)$ creates a diagonal matrix with given entries in the diagonal.

The notation $x \asymp y$ denotes that $x-y\xrightarrow[K,M_{TX}\rightarrow \infty]{a.s.}0$. The notation $1_{a=b}$ returns $1$ when $a=b$ and $0$ otherwise. The notation $\mathfrak{i}$ denotes the imaginary unit. A random vector $\xv\sim\CN(\bm\mu,\bm\Theta)$ is complex Gaussian distributed with mean vector $\bm\mu$ and covariance matrix $\bm\Theta$. The notation $\triangleq$ is used in a definition of a scalar, vector or matrix.
\section{System Model}\label{se:SM}

\subsection{Transmission Model}
We consider a communication system where $n$~TXs jointly serve $K$~RXs over a joint processing CoMP transmission network. Each TX is equipped with $M_{\TX}$~antennas, while the total number of transmit antennas is denoted by $M= n M_{\TX}$. Every RX is equipped with a single antenna. We assume that the ratio of transmit antennas with respect to the number of users is fixed and given by
\begin{equation}
\beta\triangleq \frac{M}{K}\geq 1.
\end{equation}


The signal $y_k$ received at RX~$k$ reads as
\begin{align}
y_k=\bh_k^{\He}\xv+n_k
\label{eq:SM_2.1}
\end{align}
and the overall receiving signal at all RXs is described as
\begin{align}
\yv&
=
\bH\xv
+
\bn
\label{eq:SM_2}
\end{align}
where $\yv\triangleq\begin{bmatrix}
y_1&\ldots&
y_K
\end{bmatrix}^{\Te}\in \mathbb{C}^{K\times 1}$, $\bH\triangleq \begin{bmatrix}
\bh_1&\ldots&
\bh_K
\end{bmatrix}^{\He}\in\mathbb{C}^{K\times M}$ is the CoMP channel. $\bh_k^{\He}\in\mathbb{C}^{1\times M}$ is the channel from all transmit antennas to RX~$k$. $\xv\in\mathbb{C}^{M\times 1}$ is the transmitted signal and $\bn\triangleq\begin{bmatrix}
n_1&\ldots&n_K\end{bmatrix}^{\Te}\in \mathbb{C}^{K\times 1}$ is the noise at the $K$~RXs. The transmission noise has i.i.d entry $n_k\sim\CN(0,1),\forall k=1,\ldots,K$.

The multi-user transmit signal~$\xv\in \mathbb{C}^{M\times 1}$ is obtained from the symbol vector $\bs\triangleq [s_1,\ldots,s_K]^{\Te} \in  \mathbb{C}^{K\times 1}$:
\begin{equation}
\xv=\bT \bs=\sum_{k=1}^{K}\bt_ks_k
\label{eq:SM_3}
\end{equation}
with $\bT \triangleq \begin{bmatrix}
\bt_1,
\hdots,
\bt_K
\end{bmatrix}\in \mathbb{C}^{M\times K}$ being the \emph{multi-user} precoder, $\bt_k \in \mathbb{C}^{M\times 1}$ being the beamforming vector for RX~$k$. We consider an average sum power constraint
\begin{equation}
\trace\left(\bT\bT^{\He}\right)=P,
\end{equation}
where $P$ is the average total transmit power for all TXs.

In addition, the channel to RX~$k$ is modeled as:
\begin{align}
\bh_k=\sqrt{M}\bm{\Theta}_k^{\frac{1}{2}}\bz_k
\end{align}
where $\bm{\Theta}_k\in\mathbb{C}^{M\times M}$ is the channel correlation matrix of RX$~k$ and $\bz_k$ has i.i.d complex entries of zero mean, variance $\frac{1}{M}$ and eighth order moment of order $O(\frac{1}{M^4})$. The channel correlation matrices $\bm{\Theta}_k,\forall k=1,\ldots,K$ are assumed to be slowly varying compared to the channel coherence time and therefore to be perfectly known by \emph{all} TXs.

With the assumption of Gaussian signaling~$s_k\sim\CN(0,1), \forall k$ and each user decoding with perfect CSIR, the signal-to-interference-plus-noise ratio (SINR) at RX~$k$ is given by \cite{Cover2006}
\begin{equation}
\SINR_k=\frac{\left|\bh_k^{\He}\bt_k\right|^2}{1+\overset{K}{\underset{\ell=1,\ell \neq k}{\sum}}\left|\bh_k^{\He}\bt_{\ell}\right|^2}.
\label{eq:SM_5}
\end{equation}
The ergodic sum rate for the CoMP network is then equal to
\begin{equation}
\Rate_{sum}\triangleq \sum_{k=1}^K\E\LSB \log_2\LB 1+\SINR_k\RB\RSB
\label{eq:SM_4}
\end{equation}
where the expectation is taken over the random channel realizations.
\subsection{D-CSIT Model}\label{ss:PCCSI}
Note that while we assume all TXs are endowed with a perfect copy of the user message packet to be sent on the downlink to the user terminal (e.g. user contents have been pre-routed or pre-cached at the TXs), we instead focus on the limitation of instantaneous CSI acquisition. In the D-CSIT model, each TX receives its own CSI estimate for the CoMP channel. This multi-user estimate received at the TXs is the result of feedback and CSI sharing protocols and is imperfect due to the limited resources available. The actual feedback and exchange mechanism based on which the TXs receive the multi-user channel estimate is left unspecified and arbitrary\cite{Kobayashi2011,Cheng2010}.

After this CSI sharing step, TX~$j$ acquires~$\hat{\bH}^{(j)}\triangleq \LSB\begin{array}{lll}\hat{\bh}_1^{(j)}&\ldots&\hat{\bh}_K^{(j)}\end{array}\RSB^{\He}\in \mathbb{C}^{K\times M}$ which is the multi-user channel estimate and designs its transmit coefficients \emph{without any exchange of information or iterations with the other TXs}. 

Following conventional models in the literature\cite{Jindal2006,Wagner2012,Couillet2011}, the imperfect channel estimate $\hat{\bh}_k^{(j)}$ for RX~$k$ at TX~$j$ is then modeled as
\begin{align}
\hat{\bh}^{(j)}_{k}=\sqrt{M}\bm{\Theta}_k^{\frac{1}{2}}\left(\sqrt{1-(\sigma_k^{(j)})^2}\bz_k+\sigma_k^{(j)}\bq_k^{(j)}\right)=\sqrt{1-(\sigma_k^{(j)})^2}\bh_k+\sigma_k^{(j)}\bm{\delta}_k^{(j)}.
\label{eq:SM_6}
\end{align}
The estimation error $\bm{\delta}_k^{(j)}=\sqrt{M}\bm{\Theta}_k^{\frac{1}{2}}\bq_k^{(j)}\in \mathbb{C}^{M\times 1}$, where $\bq_k^{(j)}$ has i.i.d complex entries of zero mean, variance $\frac{1}{M}$, eighth order moment of order $O(\frac{1}{M^4})$ and are independent of $\bz_k$ and $n_k$. The parameter $\sigma_k^{(j)}\in\LSB 0,1\RSB$ indicates the accuracy of the CSIT relative to the channel to RX~$k$, as seen at TX~$j$. For example, $\sigma_k^{(j)}=0$ correspond to perfect CSIT, whereas $\sigma_k^{(j)}=1$ corresponds to the channel estimate being completely uncorrelated with the true channel.

Further, we assume that the estimation errors at TX~$j$ and TX~$j'$ satisfy
\begin{align}
\bq_k^{(j)}=\rho^{(j,j')}_k\bq_k^{(j')}+\sqrt{1-(\rho^{(j,j')}_k)^2}\be^{(j,j')}_k, \forall j,j',k,
\end{align}
where $\rho^{(j,j')}_k\in\LSB 0,1\RSB$ is the correlation between $\bq_k^{(j)}$ and $\bq_k^{(j')}$. The vector $\be^{(j,j')}_k$ has i.i.d complex entries of zero mean, variance $\frac{1}{M}$, eighth order moment of order $O(\frac{1}{M^4})$ and are independent of $\bq_k^{(j')}$. Hence, the CSI estimation errors satisfy
\begin{align}
\mathbb{E}\LSB \bm{\delta}_k^{(j)}(\bm{\delta}_k^{(j')})^{\He} \RSB = \bm{\Theta}_k^{\frac{1}{2}}\mathbb{E}\LSB \bq_k^{(j)}(\bq_k^{(j')})^{\He}\RSB\bm{\Theta}_k^{\frac{\He}{2}} =\rho^{(j,j')}_k\bm{\Theta}_k.
\label{eq:correlateerror}
\end{align}
Note that $\rho^{(j,j)}_k=1, \forall j,k$.

This D-CSI model which allows for correlation between the estimate errors at different TXs is very general. It is particularly adapted to model \emph{imperfect} CSI backhaul between TXs where delay and/or imperfections are introduced.
\begin{example}
\begin{figure}[tp!]
\centering
\includegraphics[width=1\columnwidth]{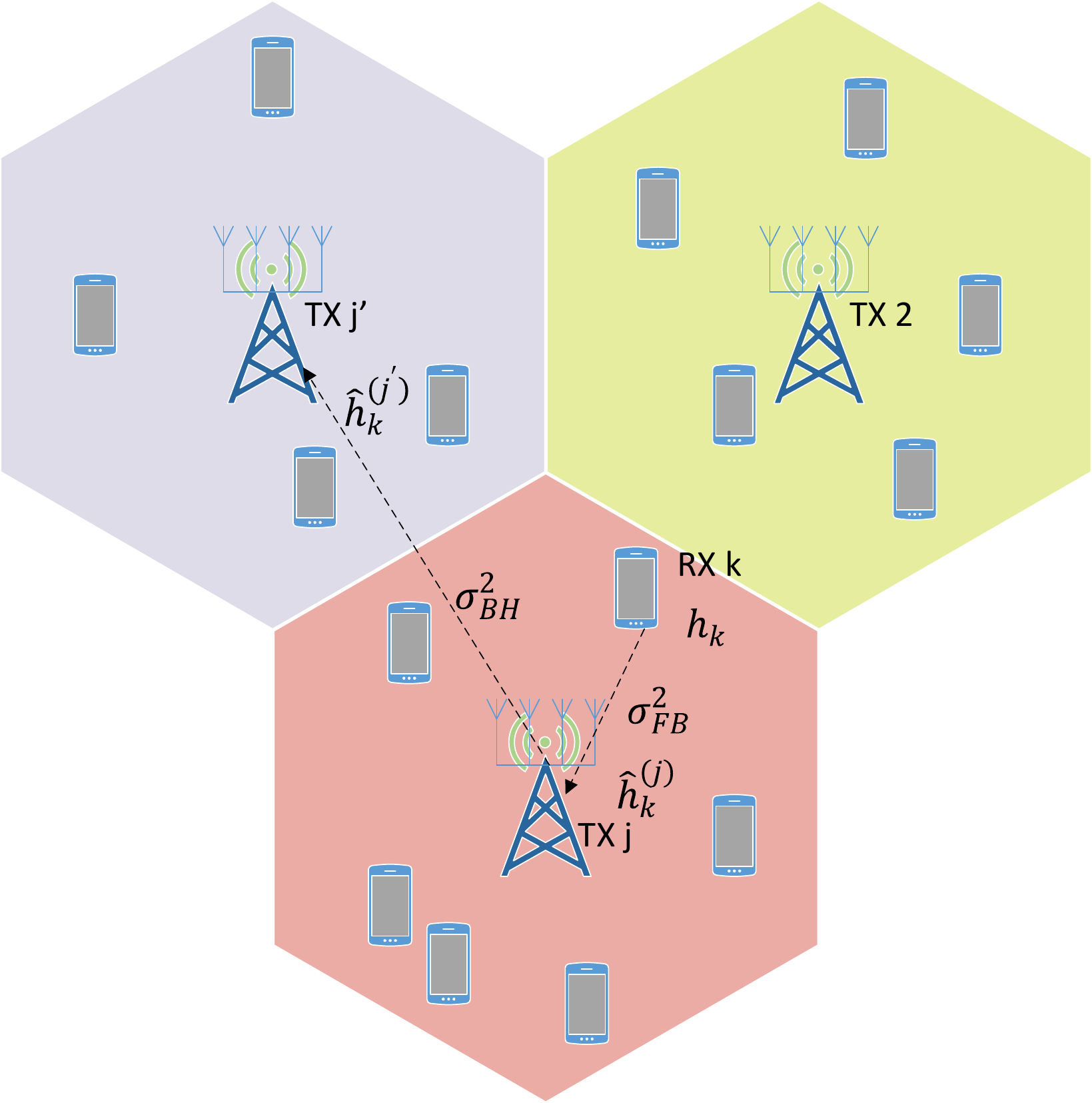}
\caption{CoMP transmission network with limited CSI feedback and limited CSI sharing}
\label{compfig}
\end{figure}
Consider a particular CoMP network setting illustrated in Fig. \ref{compfig}. In a LTE FDD downlink channel estimation scenario, each base station (TX) sends pilots to all the served users (RXs). The RX~$k$ only feedback its downlink CSI to its associated base station, the TX~$j$. The CSIT seen at TX~$j$ for RX~$k$ can then be modeled as
\begin{align*}
\hat{\bh}_k^{(j)}=\sqrt{1-\sigma_{\mathrm{FB}}^2}\bh_k+\sigma_{\mathrm{FB}}\bm\delta_k^{(j)},
\end{align*}
where $\sigma_{\mathrm{FB}}^2\in (0,1)$ parameterizes the feedback quality and $\bm\delta_k^{(j)}\sim \CN(0,1)$ is the channel independent feedback noise.

Following the LTE-architecture, this channel estimate is then shared to the other TXs through backhaul links. During this sharing step, this estimate is further degraded such that the estimate received at TX~$j'$ is written as 
\begin{align*}
\hat{\bh}_k^{(j')}&=\sqrt{1-\sigma_{\mathrm{BH}}^2}\hat{\bh}_k^{(j)}+\sigma_{\mathrm{BH}}\bm\epsilon_k^{(j,j')},\\
\end{align*}
where $\sigma_{\mathrm{BH}}\in (0,1)$ parameterizes the backhaul quality and $\bm\epsilon_k^{(j,j')}$ is the sharing noise independent from $\bh_k,\bm\delta_k^{(j)}$.

After basic algebraic operation, it can be seen that this CSIT configuration is a D-CSIT configuration with the parameters:
\begin{align}\label{calsigmauseFBBH}
\sigma_k^{(j)}&=\sigma_{\mathrm{FB}}\notag\\
\sigma_k^{(j')}&=\sqrt{1-\LB1-\sigma_{\mathrm{BH}}^2\RB\LB1-\sigma_{\mathrm{FB}}^2\RB}\notag\\
\rho_k^{(j,j')}&=\frac{\sigma_{\mathrm{FB}}\sqrt{1-\sigma_{\mathrm{BH}}^2}}{\sqrt{1-\LB1-\sigma_{\mathrm{BH}}^2\RB\LB1-\sigma_{\mathrm{FB}}^2\RB}}
\end{align}\qed
\end{example}

\begin{remark}
The D-CSIT model bridges the gap between the two extreme configuration: centralized CSIT and fully distributed CSIT. Indeed, choosing
\begin{equation}
\sigma^{(j)}_k=\sigma^{(j')}_k,~\rho^{(j,j')}_k=1, \qquad \forall j,j'\in \{1,\dots,n\},~\forall k\in \{1,\dots,K\}
\end{equation}
corresponds to the centralized CSIT configuration \cite{Jindal2006,Wagner2012}, while choosing
\begin{equation}
\rho^{(j,j')}_k=0, \qquad \forall j,j'\in \{1,\dots,n\},j\neq j',~\forall k\in \{1,\dots,K\}
\end{equation}
simplifies to the fully distributed CSIT configuration with uncorrelated estimation errors as previously studied in the literature\cite{dekerret2012_TIT}.\qed
\end{remark}

\subsection{Regularized Zero Forcing with Distributed CSI}\label{se:SM:ZF}
We consider in this work the analysis of \emph{RZF} precoder \cite{Spencer2004,Peel2005}, when faced with CSIT inconsistencies in the large system regime. Hence, the precoder designed at TX~$j$ is assumed to take the form
\begin{align}
\bT_{\rZF}^{(j)}\triangleq \LB (\hat{\bH}^{(j)})^{\He}\hat{\bH}^{(j)}+M\alpha^{(j)} \I_{M} \RB ^{-1} (\hat{\bH}^{(j)})^{\He} \frac{\sqrt{P}}{\sqrt{\Psi^{(j)}}}.
\label{eq:SM_7}
\end{align}
The scalar~$\Psi^{(j)}$ corresponds to the power normalization at TX~$j$. Hence, it holds that
\begin{align}
\Psi^{(j)}&\triangleq\norm{\LB (\hat{\bH}^{(j)})^{\He}\hat{\bH}^{(j)}+M\alpha^{(j)} \I_{M} \RB ^{-1}(\hat{\bH}^{(j)})^{\He}}^2_{\Fro}.
\label{eq:SM_10}
\end{align}
The regularization factor $\alpha^{(j)}>0,\forall j$. We also define
\begin{align}
\bC^{(j)}\triangleq \frac{ (\hat{\bH}^{(j)})^{\He}\hat{\bH}^{(j)}}{M}+\alpha^{(j)}\I_{M}.
\label{eq:SM_8}
\end{align}
Therefore, the precoder at TX~$j$ can be rewritten as
\begin{align}
\bT_{\rZF}^{(j)}=  \frac{1}{M}(\bC^{(j)})^{-1} (\hat{\bH}^{(j)})^{\He} \frac{\sqrt{P}}{\sqrt{\Psi^{(j)}}}.
\label{eq:SM_9}
\end{align}

Let $\bE_j^{\He}\in \mathbb{C}^{M_{\TX}\times M}$ denote the block selection matrix defined as
\begin{equation}
\bE_j^{\He}\triangleq \begin{bmatrix}
\bm{0}_{M_{\TX}\times (j-1)M_{\TX}}&\I_{M_{\TX}}&\bm{0}_{M_{\TX}\times (n-j)M_{\TX}}
\end{bmatrix}.
\label{eq:SM_11}
\end{equation}
Upon concatenation of all TX's precoding matrices, the effective global precoder denoted by $\bT_{\rZF}^{\DCSI}$, is written as
\begin{equation}
\bT_{\rZF}^{\DCSI}
\triangleq\begin{bmatrix}
\mu_1\bE_1^{\He}\bT_{\rZF}^{(1)}\\
\mu_2\bE_2^{\He}\bT_{\rZF}^{(2)}\\
\vdots\\
\mu_n\bE_n^{\He}\bT_{\rZF}^{(n)}
\end{bmatrix},
\label{eq:SM_10}
\end{equation}
where the scalar $\mu_j>0$ is the transmit power scaling at TX~$j$. Assume the transmit power allocated at TX~$j$ reads
\begin{align}
P_{TX_{j}}=\mu_j^2\trace\left(\bE_j\bE_j^{\He}\bT_{\rZF}^{(j)}(\bT_{\rZF}^{(j)})^{\He}\right).
\end{align}
Based on the sum power constraint,
\begin{align}\label{powerconstraint}
\sum_{j=1}^nP_{TX_{j}}=\sum_{j=1}^n\mu_j^2\trace\left(\bE_j\bE_j^{\He}\bT_{\rZF}^{(j)}(\bT_{\rZF}^{(j)})^{\He}\right)=P.
\end{align}

The finite SNR rate analysis under the precoding structure \eqref{eq:SM_10} and the D-CSIT model in \eqref{eq:SM_6} is challenging due to the dependency of each user performance on all channel estimates. Yet, some useful results can be obtained in the large antenna regime as shown below.

\section{Deterministic Equivalent of the SINR}\label{se:main}
In this section, the analysis of the so-called deterministic equivalent of the SINR under the RZF precoding is presented.

In order to derive a deterministic equivalent, we make the following standard technical assumption on the correlation matrices $\bm\Theta_k$ and the Gram matrix $\frac{1}{M}(\hat{\bH}^{(j)})^{\He}\hat{\bH}^{(j)}$ \cite{Wagner2012}.
\begin{assumption}\label{assum:corr}
All correlation matrices $\bm\Theta_k,\forall k=1,\ldots,K$ have uniformly bounded spectral norm on $M$, i.e.,
\begin{align}
\underset{M,K\rightarrow\infty}{\limsup}\sup_{1\leq k\leq K}\parallel\bm\Theta_k\parallel<\infty.
\end{align}
\end{assumption}


\begin{assumption}\label{assum:gramsnorm}
The random matrices $\frac{1}{M}(\hat{\bH}^{(j)})^{\He}\hat{\bH}^{(j)},\forall j=1,\ldots,n$ have uniformly bounded spectral norm on $M$ with probability one, i.e.,
\begin{align}
\limsup_{M,K\rightarrow\infty}\parallel\frac{1}{M}(\hat{\bH}^{(j)})^{\He}\hat{\bH}^{(j)}\parallel<\infty
\end{align}
with probability one.
\end{assumption}

Our approach will be based on the following fundamental result based on the Stieltjes transform in the analysis of wireless networks \cite{Wagner2012,Muller2013}.
\begin{theorem}{\cite{Hachem2010,Muller2013}}
\label{fundamental_theorem}
Let the matrix~$\bU$ be any matrix with bounded spectral norm and the $i$th row $\bh_i^{\He}$ of $\bH$ be $\bh_i^{\He}=\sqrt{M}\bm\Theta_i^{\frac{1}{2}}\bz_i^{\He}$, where the entries of $\bz_i$ are i.i.d of zero mean, variance $\frac{1}{M}$ and have eighth moment of order $O(\frac{1}{M^4})$. Let Assumption \ref{assum:corr} holds true. Consider the resolvent matrix $\bQ\triangleq \LB\frac{\bH^{\He}\bH}{M} +\alpha\I_M\RB^{-1}$ with regularization coefficient $\alpha>0$. Let
\begin{equation}
\bQ_o \triangleq \LB \frac{1}{M}\sum_{k=1}^{K}\frac{\bm{\Theta}_k}{1+m_k}+\alpha\bI_M\RB^{-1}
\label{eq:main_2}
\end{equation}
where $m_k$ satisfies:
\begin{equation}
m_k= \frac{1}{M}\trace\LB \bm{\Theta}_k\LB\frac{1}{M}\sum_{\ell=1}^{K}\frac{\bm{\Theta}_{\ell}}{1+m_{\ell}}+\alpha\bI_M\RB^{-1}\RB.
\label{eq:mk}
\end{equation}
Then,
\begin{equation}
\frac{1}{M}\trace \LB \bU\bQ \RB -\frac{1}{M}\trace \LB \bU\bQ_o \RB    \xrightarrow[K,M\rightarrow \infty]{a.s.}0.
\label{eq:main_3}
\end{equation}
\end{theorem}
The fixed point~$m_k$ can easily be obtained by an iterative fixed-point algorithm described in \cite{Wagner2012,Couillet2011} and recalled in Appendix~\ref{app:literature} for the sake of completeness.

Adopting the shorthand notation used in \cite{Wagner2012}, we introduce
\begin{equation}\label{definec}
c_{0,k}^{(j)}\triangleq 1-(\sigma^{(j)}_k)^2\!, \quad
c_{1,k}^{(j)}\triangleq (\sigma^{(j)}_k)^2\!, \quad
c_{2,k}^{(j)}\triangleq \sigma^{(j)}_k\sqrt{1\!-\!(\sigma^{(j)}_k)^2}.
\end{equation}
We can further define the term $\bQ_o^{(j)}$ and $m_{k}^{(j)}$ respectively as $\bQ_o$ and $m_{k}$ in Theorem \ref{fundamental_theorem} using instead the local CSI estimate $\hat{\bH}^{(j)}$ and regularization coefficient $\alpha^{(j)}$ at TX~$j$. A deterministic equivalent of the SINR under RZF precoding is therefore provided in the following theorem.
\begin{theorem}
\label{theorem}
Let the Assumptions \ref{assum:corr} and \ref{assum:gramsnorm} hold true, then the SINR of RX~$k$ under RZF precoding satisfies
\begin{align}
\SINR_k-\SINR_k^o\xrightarrow[K,M_{TX}\rightarrow \infty]{a.s.} 0
\end{align}
with $\SINR_k^o$ defined as
\begin{align}
\SINR_k^o\triangleq \frac{P\LB \sum_{j=1}^n  \mu_j\sqrt{\frac{c_{0,k}^{(j)}}{\Gamma^o_{j,j}(\bI_M)}} \frac{ \Phi^{o}_{j,k}}{1+m_{k}^{(j)}}\RB^2}{1+I_k^o}
\label{eq:main_4}
\end{align}
with~$I_k^o\in\mathbb{R}$ given by
\begin{align}
&I_k^o\!\triangleq P\sum_{j=1}^n\sum_{j'=1}^n \frac{\mu_j\mu_{j'}}{\sqrt{\Gamma^o_{j,j}(\bI_M)\Gamma^o_{j',j'}(\bI_M)}}\left(\Gamma^o_{j,j'}(\bE_{j'}\bE_{j'}^{\He}\bm{\Theta}_k\bE_j\bE_j^{\He})-2\Gamma^o_{j,j'}(\bm{\Theta}_k\bE_j\bE_j^{\He})\frac{  c_{0,k}^{(j')}\Phi^{o}_{j',k}}{1+m^{(j')}_{k}}\right.\notag\\
&\left.+\Phi^{o}_{j',k}\Phi^{o}_{j,k}\Gamma^o_{j,j'}(\bm\Theta_k)\frac{c_{0,k}^{(j)}c_{0,k}^{(j')}+\rho_k^{(j,j')}c_{2,k}^{(j)}c_{2,k}^{(j')}}{(1+m_k^{(j)})(1+m_k^{(j')})}\right).
\label{eq:main_5}
\end{align}
where $\Phi^{o}_{j,k}\in\mathbb{R}$ is defined as
\begin{align}
\Phi^{o}_{j,k}=\frac{\trace\LB\bm{\Theta}_k\bE_{j}\bE_{j}^{\He}\bQ_o^{(j)}\RB}{M},
\end{align}
and the function $\Gamma_{j,j'}^o(\bX):\mathbb{C}^{M\times M}\mapsto \mathbb{C}$ is defined in Lemma \ref{lemma1}.
The transmit power scaling $\mu_j$ for TX~$j$ satisfies
\begin{align}\label{powerscalingmu}
\sum_{j=1}^n\mu_j^2\frac{\Gamma_{j,j}^{o}(\bE_j\bE_j^{\He})}{\Gamma_{j,j}^{o}(\bI_M)}=1
 \end{align}
\end{theorem}

\begin{proof}
The proof of Theorem \ref{theorem} is given in Appendix \ref{se:proof}.
\end{proof}

The theorem demonstrates that in the large system setting, the SINR expression for each RX can be derived as a given function of (i) $n,M_{TX},K$ that indicate the system dimensions, (ii) $\sigma_k^{(j)},\rho_k^{(j,j')},\bm\Theta_k$ which reflect the statistics of the channel and of CSI estimates at each TX, and (iii) the precoder regularization coefficients $\alpha^{(j)}$ and power scalings $\mu^{(j)}$.

This result is very general and encompasses several important results from the literature.

\subsection{Regularized ZF Precoding for Centralized CSI Isotropic Channel}\label{ss:CCSI}

Choosing $\sigma^{(j)}_k=\sigma^{(j')}_k=\sigma_k$, $\alpha^{(j)}=\alpha^{(j')}=\alpha$, $\rho^{(j,j')}_k=1,\forall j,j'\in \{1,\dots,n\},k\in \{1,\dots,K\}$, we obtain the centralized CSIT configuration. Further assuming that $\bm\Theta_k=\bI_M$, $m_k^{(j)}$ is obtained in closed form as
\begin{align}\label{closedformmkj}
m_k^{(j)}=m^{o}=\frac{\beta-1-\alpha\beta+\sqrt{(\alpha\beta-\beta+1)^2+4\alpha\beta^2}}{2\alpha\beta}.
\end{align}
In this setting, the total power constraint \eqref{powerscalingmu} simplifies to
\begin{align}
\frac{1}{n}\sum_{j=1}^n\mu_j^2=1
 \end{align}
since
\begin{align}
\Gamma^o_{j,j}(\I_M)&=\frac{(m^{o})^2}{\beta(1+m^{o})^2-(m^{o})^2},\\
\Gamma^o_{j,j}(\bE_j\bE_j^{\He})&=\frac{1}{n}\Gamma^o_{j,j}(\I_M).
\end{align}

Assume $\mu_j=1,\forall j=1,\ldots,n$, the transmit power $p_{TX_j}$ at TX~$j$ denotes
\begin{align}
p_{TX_j}=\mu_j^2P\frac{\Gamma^o_{j,j}(\bE_j\bE_j^{\He})}{\Gamma^o_{j,j}(\I_M)}=\frac{P}{n}
\end{align}
This indicates an equal power allocation per TX.
Since $\bm\Theta_k=\bI_M$ the channel is isotropic, the above setting also indicates that the signal power for RX~$k$ satisfies
\begin{align}
p_k=\frac{P}{K}
\end{align}
which is an equal power per RX.

After simple algebraic manipulations, we can obtain the deterministic equivalent of SINR in \eqref{eq:main_4}
\begin{align}
\SINR^o_k&=\frac{(1-\sigma_k^2)(\beta \LB 1+ m^{o}\RB^2-(m^{o})^2)}{\left(1-\sigma_k^2+(1+m^{o})^2\sigma_k^2+\frac{(1+m^{o})^2}{P}\right)}
\end{align}
This coincides with the results in \cite[Corollary $2$]{Wagner2012}.

\subsection{Regularized ZF Precoding for Fully Distributed CSI Isotropic Channel}

Choosing $\rho^{(j,j')}_k=0,\forall~j,j'\in \{1,\dots,n\},j\neq j', k\in \{1,\dots,K\}$, the fully distributed CSIT configuration with uncorrelated estimation errors is obtained. Let us further assume that the same regularization coefficient is used at each TX, i.e., $\alpha^{(j)}=\alpha^{(j')}=\alpha,\forall~j,j'\in \{1,\dots,n\}$, $\bm\Theta_k=\bI_M$ and $\mu_j=1$ indicating equal per TX power allocation.

The deterministic SINR in \eqref{eq:main_4} then becomes 
\begin{align}
\SINR_k^{o}=\frac{P\left(\frac{1}{n}\sum_{j=1}^n  \sqrt{c_{0,k}^{(j)}}\right)^2 \frac{\beta \LB 1+ m^{o}\RB^2-(m^{o})^2}{(1+m^{o})^2}}{I_k^{o}+1}
\end{align}
with
\begin{align}
I_k^{o}&= P-P\sum_{j=1}^n\sum_{j'=1}^n\frac{\left(\beta \LB 1+ m^{o}\RB^2-(m^{o})^2\right)\Gamma^o_{j,j'}}{n^2(1+m^{o})^2m^{o}}\cdot\left[2c_{0,k}^{(j)}
+m^{o}\LB 2c_{0,k}^{(j)}-c_{0,k}^{(j)}c_{0,k}^{(j')}\RB\right]\\
\Gamma_{j,j'}^o&= \! \frac {\frac{1}{M}\sum_{\ell=1}^K \sqrt{c_{0,\ell}^{(j)}c_{0,\ell}^{(j')}}}{ \frac{(1+m^{o})^2}{(m^{o})^2}- \frac{1}{M}\sum_{\ell=1}^K\!c_{0,\ell}^{(j)}c_{0,\ell}^{(j')}\!}
\end{align}
This result coincides with \cite{dekerret2015_ISIT}.

\subsection{Regularized ZF Precoding for D-CSI Isotropic Channel}

Assume that $\bm\Theta_k=\bI_M,\forall k\in \{1,\ldots,K\}$ and $\mu_j=1,\forall j\in \{1,\ldots,n\}$, indicating equal per TX power allocation. In this specific setting, the terms $m_k^{(j)}$ can be obtained in closed form as
\begin{align}
m_k^{(j)}=m^{(j)}=\frac{\beta-1-\alpha^{(j)}\beta+\sqrt{(\alpha^{(j)}\beta-\beta+1)^2+4\alpha^{(j)}\beta^2}}{2\alpha^{(j)}\beta}.
\end{align}
After simplification, the deterministic SINR in \eqref{eq:main_4} becomes
\begin{align}\label{allertonSINR}
\SINR_k^o =\frac{P\LB\frac{1}{n}\sum_{j=1}^n  \sqrt{\frac{1-(\sigma^{(j)}_k)^2}{\Gamma_{j,j}^o}} \frac{m^{(j)}}{1+m^{(j)}}\RB^2}{1+I_k^o}
\end{align}
with~$I_k^o\in\mathbb{R}$ defined as
\begin{align}
&I_k^o\!= P
-P\sum_{j=1}^n\sum_{j'=1}^n\frac{\Gamma_{j,j'}^o}{\sqrt{\Gamma_{j,j}^o\Gamma_{j',j'}^o}}\left[\frac{2c_{0,k}^{(j)}}{n^2}\frac{m^{(j)}}{1+m^{(j)}}
-\frac{\LB\rho^{(j,j')}_kc_{2,k}^{(j)}c_{2,k}^{(j')}+c_{0,k}^{(j)}c_{0,k}^{(j')}\RB m^{(j)}m^{(j')}}{n^2\LB1\!+\!m^{(j)}\RB\LB1\!+\!m^{(j')}\RB}\right]
\label{eq:main_5}
\end{align}
where $\Gamma^o_{j,j'}\in\mathbb{R}$ is given by
\begin{align}
\Gamma_{j,j'}^o&\!= \! \frac {\frac{1}{M}\sum_{\ell=1}^K \sqrt{c_{0,\ell}^{(j)}c_{0,\ell}^{(j')}}\!+\!\sqrt{c_{1,\ell}^{(j)}c_{1,\ell}^{(j')}}\rho^{(j,j')}_{\ell}}{ \frac{1+m^{(j)}}{m^{(j)}}\frac{1+m^{(j')}}{m^{(j')}}- \frac{\sum_{\ell=1}^K\!\LB\!\sqrt{c_{0,\ell}^{(j)}c_{0,\ell}^{(j')}}\!+\!\sqrt{c_{1,\ell}^{(j)}c_{1,\ell}^{(j')}}\rho^{(j,j')}_{\ell}\!\RB^2\!}{M}\!}
\end{align}
This result coincides with \cite{Li2015_Allerton}.

\section{Applications of the Theorem}
The deterministic equivalent of the SINR expression allows to evaluate the performance of RZF precoding. However, there is an added benefit here, which is the possibility to optimize the transmission parameters (i.e., regularization coefficient) so as to obtain some robustness with respect to the D-CSIT configuration, as it will be discussed in the following.

\subsection{Robust Sum Rate Maximizing Regularization}\label{ss:RSRMR}
If there exists a predefined per TX power constraint such that the average transmit power for each TX~$j$ is given as $P_{TX_{j}}=p_j$, according to Theorem \ref{theorem}, we can find the power scaling parameter for each TX~$j$ as
\begin{align}\label{expformu}
\mu_j=\sqrt{\frac{p_j\Gamma_{j,j}^{o}(\bI_M)}{P\Gamma_{j,j}^{o}(\bE_j\bE_j^{\He})}}.
\end{align}
Substituting \eqref{expformu} into Theorem \ref{theorem}, the ergodic sum rate becomes a function only depending on $\alpha^{(j)},j=1,\ldots,n$.
\subsubsection{Robust Regularized ZF}\label{sss:RRZF}
The regularization coefficients tuple $\bm\alpha=\LSB\alpha^{(1)},\ldots,\alpha^{(n)}\RSB$ which maximizes the system sum rate while being robust to the D-CSIT configuration is given by
\begin{align}\label{optimizeprobalpha}
\bm\alpha^{\star}\triangleq\argmax_{\bm\alpha} \sum_{k=1}^{K}\log\LB1+\SINR_k^o\RB,s.t.~\mu_j=\sqrt{\frac{p_j\Gamma_{j,j}^{o}(\bI_M)}{P\Gamma_{j,j}^{o}(\bE_j\bE_j^{\He})}},\forall j.
\end{align} 

\subsubsection{Robust regularized ZF with equal regularization}\label{sss:RRZFeqreg}
The problem \eqref{optimizeprobalpha} is a non-convex optimization. In order to reduce the complexity, we introduce the following optimization assuming that the regularization coefficients are the same at different TXs.
\begin{align}\label{optimizeprobsamealpha}
\alpha_{same}^{\star}\triangleq\argmax_{
\alpha_{same}} \sum_{k=1}^{K}\log\LB1+\SINR_k^o\RB,s.t.~\mu_j=\sqrt{\frac{p_j\Gamma_{j,j}^{o}(\bI_M)}{P\Gamma_{j,j}^{o}(\bE_j\bE_j^{\He})}},\forall j.
\end{align}
The optimization variable is now a scalar parameter and the global optimal regularization can be easily found using a line search algorithm \cite{wachter2005line}.

\subsubsection{Naive Regularized ZF}\label{sss:NRZF}
We introduced in the following the naive regularization optimization which doesn't take into account the D-CSIT configuration. This is therefore the reference baseline for our improved robust precoding scheme.

When TXs are not aware of the D-CSIT structure, each TX will choose its regularization parameter on the basis of its own CSI quality. This yields a naive (suboptimal) precoding scheme. Specifically, assuming equal power allocation at each TX, each TX~$j$ optimizes its regularization coefficient $\alpha^{(j)}$ based on $\hat{\bH}^{(j)}$ considering as if $\hat{\bH}^{(j)}$ is the centralized CSIT shared among all TXs, i.e.,
\begin{align}\label{NALF}
\alpha_{naive}^{(j)}\triangleq\argmax_{\alpha^{(j)}} \Rate_{sum}\LB\hat{\bH}^{(j)},\ldots,\hat{\bH}^{(j)}\RB.
\end{align}

In the particular case where the CSIT quality is homogeneous across users, i.e., $\sigma_k^{(j)}=\sigma^{(j)}, \forall k\in\{1,\ldots, K\}$, and the channel is isotropic, i.e., $\bm\Theta_k=\bI_M$, the optimal naive regularization coefficient is obtained in closed form \cite{Wagner2012}
\begin{equation}\label{alphaCCSI}
\alpha_{naive}^{(j)}=\frac{1+(\sigma^{(j)})^2P}{1-(\sigma^{(j)})^2}\frac{1}{\beta P}.
\end{equation}

\subsection{Robust Power Optimization}\label{ss:optpower}
If the regularization coefficient at each TX is predefined, according to Theorem \ref{theorem}, we can optimize the power scaling tuple $\bm\mu=\LSB\mu_1,\ldots,\mu_n\RSB$ that maximizes the system sum rate: 
\begin{align}\label{optimizemu}
\begin{array}{ll}
\bm\mu^{\star}=\argmax_{
\bm\mu}\sum_{k=1}^{K}\log\LB1+\SINR_k^o\RB, s.t.~\sum_{j=1}^n\mu_j^2\frac{\Gamma_{j,j}^{o}(\bE_j\bE_j^{\He})}{\Gamma_{j,j}^{o}(\bI_M)}=1
\end{array}.
\end{align}
Problem (\ref{optimizemu}) can then be reformulated as:
\begin{align}\label{optimizeprobalphareform}\tag{P1}
\begin{array}{ll}
\bm\mu^{\star}=\argmin_{\bm\mu}\overset{K}{\underset{k=1}{\prod}}\frac{\frac{1}{P}+\bm\mu^{\Te}\bB_k\bm\mu}{\frac{1}{P}+\bm\mu^{\Te}(\bA_k+\bB_k)\bm\mu}, s.t.~\parallel\bC\bm\mu\parallel_F^2=1,\bm\mu\in\mathbb{R}^{n}
\end{array},
\end{align}
where $\bA_k,\bB_k,\bC,\forall k$ are constant matrices defined as
\begin{align}
\LSB\bA_k\RSB_{j,j'}&\triangleq\sqrt{\frac{c_{0,k}^{(j)}c_{0,k}^{(j')}}{\Gamma^o_{j,j}(\bI_M)\Gamma^o_{j',j'}(\bI_M)}} \frac{ \Phi^{o}_{j,k}\Phi^{o}_{j',k}}{\LB1+m_{k}^{(j)}\RB\LB1+m_{k}^{(j')}\RB}\\
\LSB\bB_k\RSB_{j,j'}&\triangleq\frac{1}{\sqrt{\Gamma^o_{j,j}(\bI_M)\Gamma^o_{j',j'}(\bI_M)}}\left(\Gamma^o_{j,j'}(\bE_{j'}\bE_{j'}^{\He}\bm{\Theta}_k\bE_j\bE_j^{\He})\vphantom{-2\Gamma^o_{j,j'}(\bm{\Theta}_k\bE_j\bE_j^{\He})\frac{  c_{0,k}^{(j')}\Phi^{o}_{j',k}}{1+m^{(j')}_{k}}}\right.\notag\\
-2\Gamma^o_{j,j'}&\left.(\bm{\Theta}_k\bE_j\bE_j^{\He})\frac{  c_{0,k}^{(j')}\Phi^{o}_{j',k}}{1+m^{(j')}_{k}}+\Phi^{o}_{j',k}\Phi^{o}_{j,k}\Gamma^o_{j,j'}(\bm\Theta_k)\frac{c_{0,k}^{(j)}c_{0,k}^{(j')}+\rho_k^{(j,j')}c_{2,k}^{(j)}c_{2,k}^{(j')}}{(1+m_k^{(j)})(1+m_k^{(j')})}\right)\\
\bC&\triangleq\diag\LB\sqrt{\frac{\Gamma_{1,1}^{o}(\bE_1\bE_1^{\He})}{\Gamma_{1,1}^{o}(\bI_M)}},\ldots,\sqrt{\frac{\Gamma_{n,n}^{o}(\bE_n\bE_n^{\He})}{\Gamma_{n,n}^{o}(\bI_M)}}\RB.
\end{align}
Let $u_i(\bm\mu)$ be denoted as
\begin{align}\label{defu}
u_i(\bm\mu)\triangleq\frac{\frac{1}{P}+\bm\mu^{\Te}\bB_i\bm\mu}{\frac{1}{P}+\bm\mu^{\Te}\LB\bA_i+\bB_i\RB\bm\mu}.
\end{align}
In order to solve problem \ref{optimizeprobalphareform}, we first introduce the following lemma: 
\begin{lemma}[Adapted from Lemma 1,\cite{bogale2012weighted}]\label{lemma_parametric}
The optimal point of optimization problem
\begin{align*}
\begin{array}{ll}
\min_{\bm\mu}&\prod_{i=1}^Ku_i(\bm\mu)\\
s.t.& \parallel\bC\bm\mu\parallel_F^2=1,\bm\mu\in\mathbb{R}^{n}
\end{array}
\end{align*}
can be obtained by solving the following parametric problem
\begin{align*}
\begin{array}{ll}
\min_{\bm\mu,\{\lambda_i\}_{i=1}^K} &\LB\sum_{i=1}^K\frac{1}{K}\lambda_iu_i(\bm\mu)\RB^{\frac{1}{K}}\\
s.t.&\begin{array}{l}
\prod_{i=1}^K\lambda_i=1,\lambda_i\geq 0\\
\parallel\bC\bm\mu\parallel_F^2=1,\bm\mu\in\mathbb{R}^{n}
\end{array}
\end{array}
\end{align*}
Moreover, for fixed $\bm\mu$, the optimal $\{\lambda_i\}_{i=1}^K$ of this problem is given by
\begin{align}\label{expression_lambda}
\lambda_i^*=\frac{\LSB\prod_{\ell=1}^Ku_{\ell}(\bm\mu)\RSB^{\frac{1}{K}}}{u_{i}(\bm\mu)},\forall i.
\end{align}
\end{lemma}
\begin{remark}
The above lemma is exactly Lemma 1 presented in \cite{bogale2012weighted} with $\tilde{\xi}_s,\nu_s$ replaced by $u_{i},\lambda_i$ in order have consistent notation. The expression for $\lambda_i^*$ with fixed $\bm\mu$ is reminiscent of the expression for $\nu_s^*$ with fixed $\bb_s$.\qed
\end{remark}
According to Lemma \ref{lemma_parametric}, with some simplifications, problem \ref{optimizeprobalphareform} can be solved by the following parametric problem
\begin{align}\label{CCPlamda}\tag{P2}
\begin{array}{ll}
\underset{\bm\mu,\{\lambda_i\}_{i=1}^K}{\min}&\overset{K}{\underset{k=1}{\sum}}\lambda_k\frac{\frac{1}{P}+\bm\mu^{\Te}\bB_k\bm\mu}{\frac{1}{P}+\bm\mu^{\Te}\LB\bA_k+\bB_k\RB\bm\mu}\\
s.t. & \begin{array}{l} \prod_{i=1}^K\lambda_i=1\\ \parallel\bC\bm\mu\parallel_F^2=1,\bm\mu\in\mathbb{R}^{n\times1}.\end{array}
\end{array}
\end{align}
We hereby introduce an iterative procedure to calculate the local optimal solution for problem \ref{CCPlamda}.
\begin{algorithm}[htb]
\caption{Iterative algorithm for problem \ref{optimizeprobalphareform}}
\begin{algorithmic}[1]\label{iterativeCCP}
\STATE{Initialize $\bm\mu^{[0]}$}
\STATE{$t=0$}
\WHILE{not converge}
\STATE{$\lambda_i^{[t+1]}=\LSB\prod_{\ell=1}^K\frac{\frac{1}{P}+(\bm\mu^{[t]})^{\Te}\bB_{\ell}\bm\mu^{[t]}}{\frac{1}{P}+(\bm\mu^{[t]})^{\Te}\LB\bA_{\ell}+\bB_{\ell}\RB\bm\mu^{[t]}}\RSB^{\frac{1}{K}}\cdot\frac{\frac{1}{P}+(\bm\mu^{[t]})^{\Te}\LB\bA_i+\bB_i\RB\bm\mu^{[t]}}{\frac{1}{P}+(\bm\mu^{[t]})^{\Te}\bB_i\bm\mu^{[t]}},\qquad \forall i=1,\ldots,K$}
\STATE{$\bm\mu^{[t+1]}=\underset{\xv}{\argmin}\sum_{k=1}^K\lambda_k^{[t+1]}\cdot\frac{\frac{1}{P}+\xv^{\Te}\bB_k\xv}{\frac{1}{P}+\xv^{\Te}(\bA_k+\bB_k)\xv}, \textrm{s.t.}\parallel \bC\xv\parallel^2\leq1$}
\STATE{$t=t+1$}
\ENDWHILE
\end{algorithmic}
\end{algorithm}

The iterative optimization step in Algorithm \ref{iterativeCCP} is equivalent to a maximization for the sum of ratios of two convex functions over a convex set. It can be solved for example by a branch and bound algorithm described in \cite{shen2013maximizing}.

\begin{theorem}
\label{iterationmethod}
Algorithm \ref{iterativeCCP} converges to a local optimum of the optimization problem \ref{CCPlamda}.
\end{theorem}
\begin{proof}
The proof of Theorem \ref{iterationmethod} is given in Appendix \ref{se:proof2}.
\end{proof}
Therefore, we can obtain a local optimal power allocation such that the system sum rate is maximized under the D-CSI configuration.

\subsection{Robust Joint Optimization of Regularization and Power}\label{ss:jointopt}
In subsection \ref{ss:RSRMR} and \ref{ss:optpower}, we tackle the problem of finding the regularization coefficient (power scaling factor) which maximizes the system sum rate while the power scaling factor (regularization coefficient) is fixed, respectively. Indeed in the D-CSIT configuration, the regularization tuple $\bm\alpha$ and the power scaling tuple $\bm\mu$ can be jointly optimized according to a predefined power constraint. However, since the joint optimization for $(\bm\alpha,\bm\mu)$ is a complicated non-convex problem, we then consider an alternating optimization approach which iterates between the optimization of $\bm\alpha$ and $\bm\mu$ described in subsection \ref{ss:RSRMR} and \ref{ss:optpower}. A local optimal point can be reached while applying the alternating optimization.

In this subsection, we mainly consider two catogories of joint optimization problems described in the sequel.
 
\subsubsection{Robust Joint Optimization}
\begin{align}\label{RJO}
\LB\bm\alpha^{\star},\bm\mu^{\star}\RB=\argmax_{
\bm\alpha,\bm\mu} \sum_{k=1}^{K}\log\LB1+\SINR_k^o\RB, s.t.\sum_{j=1}^n\mu_j^2\frac{\Gamma_{j,j}^{o}(\bE_j\bE_j^{\He})}{\Gamma_{j,j}^{o}(\bI_M)}=1.
\end{align}
This corresponds to the optimal solution where both parameters are jointly optimized.
\subsubsection{Robust Joint Optimization with equal regularization}

\begin{align}\label{RJOEQR}
\LB\alpha^{\star}_{same},\bm\mu^{\star}\RB=\argmax_{
\alpha_{same},\bm\mu} \sum_{k=1}^{K}\log\LB1+\SINR_k^o\RB, s.t.\sum_{j=1}^n\mu_j^2\frac{\Gamma_{j,j}^{o}(\bE_j\bE_j^{\He})}{\Gamma_{j,j}^{o}(\bI_M)}=1.
\end{align}
This corresponds to a jointly optimization for regularization and power scaling, assuming that the regularization coefficient at all TXs are the same.
\section{Simulation Results}\label{se:simulations}
In the following, we provide simulations results to evaluate the accuracy of the deterministic expressions provided and to gain insights into the system design. We also simulate the sum rate performance of the optimal regularization coefficients and power allocation which are robust to the D-CSIT setting.
 
For the sake of conciseness, the following simulations consider an isotropic channel setting listed in Table \ref{tableiso}. Similar results can be obtained with cellular setting. 

\begin{table}[h]
\begin{center}
\begin{tabular}{|c|c|c|c|c|c|c||c|c|c|c|}
    \hline
    $M$  &  $K$ & $n$ & $\beta$ & $\bm\Theta_k$ & $P$ & $\rho_k^{(j,j')}$ & & $(\sigma_k^{(1)})^2$ & $(\sigma_k^{(2)})^2$ & $(\sigma_k^{(3)})^2$\\ \hline
    $30$ & $30$ & $3$ & $1$ & $\bI_M$ & $20$dB & \begin{tabular}{cc}fully distributed CSIT & 0 \\D-CSIT & 0.81\\centralized CSIT & 1\end{tabular} &\begin{tabular}{c} asymmetric\\symmetric\end{tabular} &\begin{tabular}{c} 0.01\\0.1\end{tabular} &\begin{tabular}{c} 0.16\\0.1\end{tabular} & \begin{tabular}{c} 0.49\\0.1\end{tabular} \\ \hline
  \end{tabular}
  \caption{Simulation parameters for the isotropic channel setting.}\label{tableiso}
\end{center}
\end{table}

In isotropic channel setting, by increasing the value of $\rho_k^{(j,j')}$ from $0$ to $1$, the CSIT structure for the system gradually changes from fully distributed CSIT to centralized CSIT. For the CSIT discrepancy at different TXs, we consider two cases in the isotropic channel setting: the asymmetric setting where the CSIT accuracy at different TXs are different and the symmetric setting where the CSIT accuracy at different TXs are the same.

In the following simulations of robust regularization and power optimization, we compare the sum rate performance of following algorithms:
\begin{itemize}
\item $\LB\bm\alpha_{naive}, \bm\mu_{eq}\RB$: A naive algorithm to obtain the regularization coefficients without considering the D-CSIT configuration, equal power is allocated at each TX (See \eqref{NALF}).
\item $\LB\alpha_{same}^{\star}, \bm\mu_{eq}\RB$: A robust optimization of regularization imposing that all TXs have the same regularization coefficient, equal power allocation is assumed at each TX (See \eqref{optimizeprobsamealpha}).
\item $\LB\bm\alpha^{\star}, \bm\mu_{eq}\RB$: A robust optimization of regularization with equal power allocation at each TX (See \eqref{optimizeprobalpha}).
\item $\LB\alpha_{same}^{\star}, \bm\mu^{\star}\RB$: A robust joint optimization of regularization and power at each TX under D-CSIT scenario, with the additional constraint that all TXs have the same regularization coefficient is imposed (See \eqref{RJOEQR}).
\item $\LB\bm\alpha^{\star}, \bm\mu^{\star}\RB$: A robust joint optimization of regularization and power at each TX under D-CSIT scenario (See \eqref{RJO}).
\end{itemize}

\subsubsection{Monte-Carlo Simulations of Theorem~\ref{theorem}}

We verify using Monte-Carlo (MC) simulations the accuracy of the asymptotic expression derived in Theorem~\ref{theorem}.

Fig.~\ref{figconvergence} depicts the absolute error of the deterministic equivalent $R_{sum}^0$ compared to the ergodic sum rate $R_{sum}$ as a function of the number of users~$K$.  The ergodic sum rate is averaged over $1000$ independent channel realizations. For ease of illustration, we choose the symmetric CSIT configuration and an equal available power per TX. Furthermore, the regularization coefficient at each TX~$j$ is chosen as $\alpha^{(j)}=\frac{1}{\beta P}$.

\begin{figure}[htp]
\centering
\includegraphics[width=1\columnwidth]{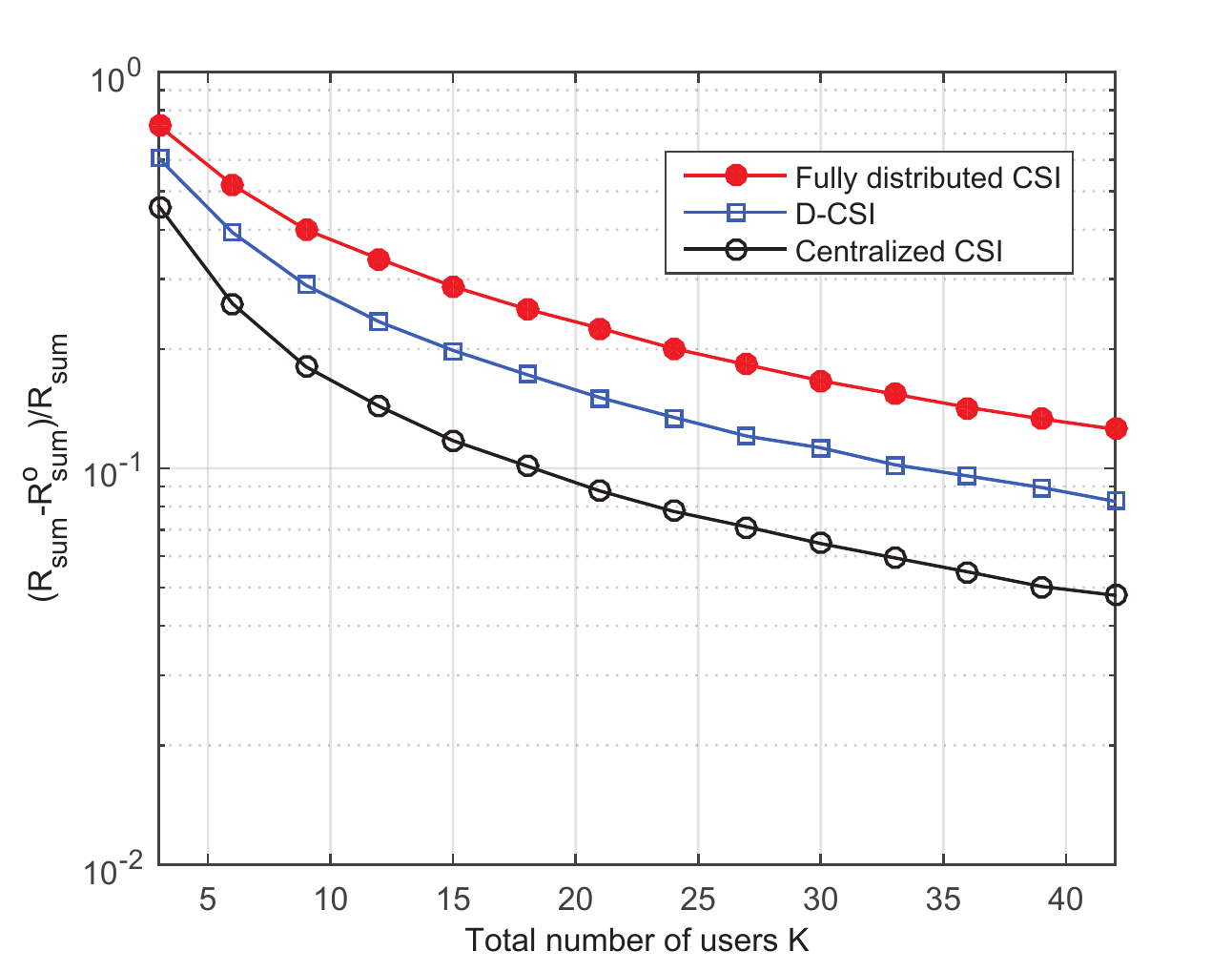}
\caption{Relative deviation between the deterministic equivalent and the Monte-Carlo simulations as a function of the number of users~$K$.}
\label{figconvergence}
\end{figure}
It can be seen that the deterministic equivalent converges to the expected sum rate obtained using Monte-Carlo simulations as the system becomes large. It also reveals that the rate of convergence is faster when the CSIT configuration becomes more centralized (i.e., when the CSIT noise becomes more correlated).
\FloatBarrier

\subsubsection{Cost of CSIT Distributiveness}

As is mention in Section \ref{ss:PCCSI}, the CSI estimate noise correlation parameter $\rho_k^{(j,j')}$ reflects the distributiveness of this CoMP network. Let us consider the symmetric accuracy setting, let the CSI estimate noise correlation be $\rho_k^{(j,j')}=\rho,\forall k,\forall j\neq j'$, we plot the ergodic sum rate when the CSI estimate noise correlation $\rho$ varies from $0$ to $1$, namely, when the CSI structure varies from fully distributed CSI to D-CSI and finally becomes centralized CSI.

\begin{figure}[htp]
\centering
\includegraphics[width=1\columnwidth]{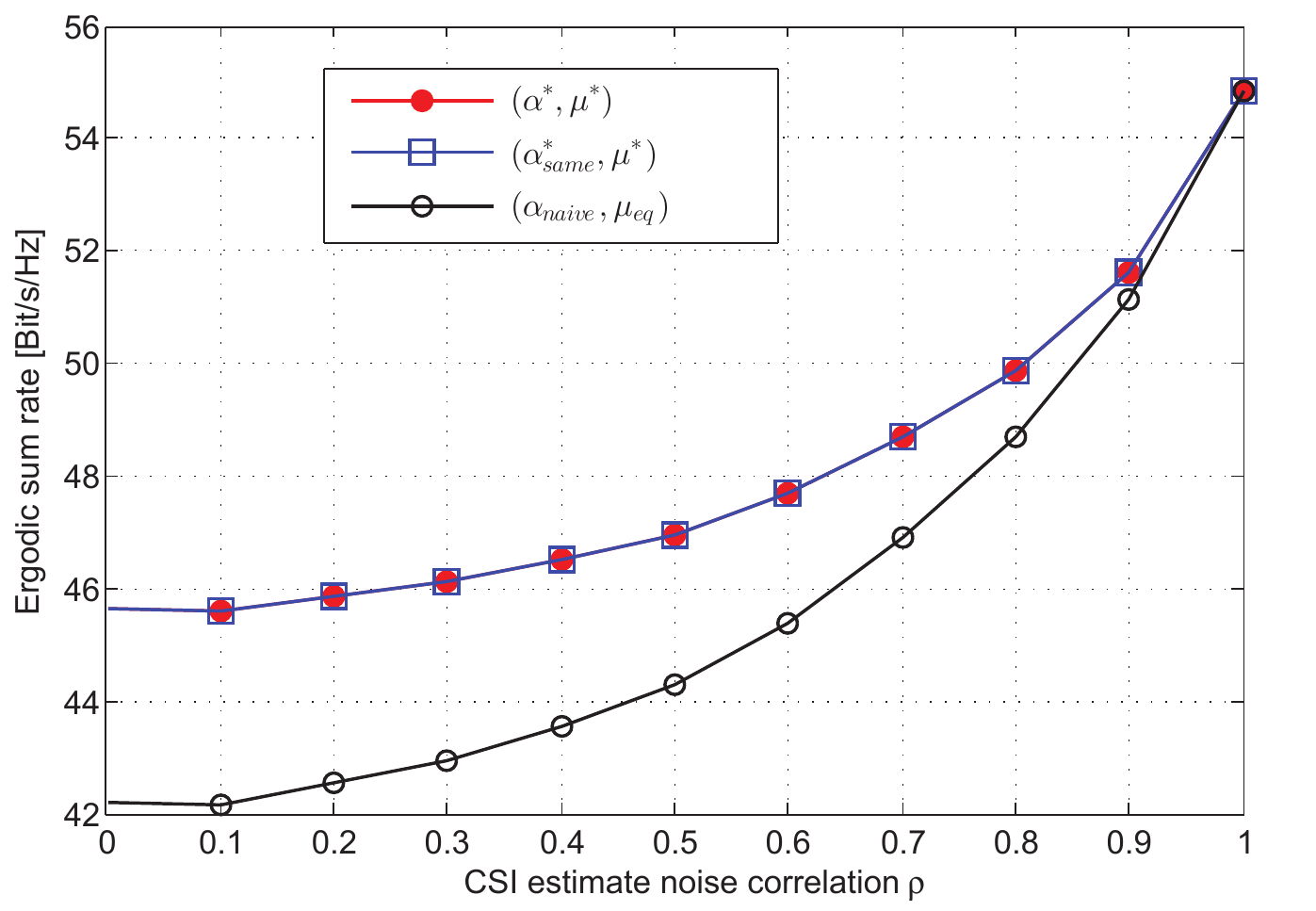}
\caption{Ergodic sum rate as a function of estimate noise correlation $\rho$ which indicates the distributiveness for the CSIT, RZF precoding is implemented.}
\label{figsrvsrho}
\end{figure}

Fig. \ref{figsrvsrho} reveals that the proposed algorithms outperforms the naive one in the D-CSI scenarios. We can also verify that the D-CSI structure introduces a non-vanishing performance degradation compared to the centralized CSI case. We can also observe that the sum rate for $(\bm\alpha^{\star},\bm\mu^{\star})$ and $(\alpha^{\star}_{same},\bm\mu^{\star})$ are very close to each other.

\FloatBarrier

\subsubsection{Joint Optimization of Regularization and Power for Isotropic Channel}

Let us consider the D-CSIT configuration with asymmetric CSIT accuracy. We then plot the ergodic sum rate as a function of the total transmit power $P$ varies from $0$ dB to $30$dB.
\begin{figure}[htp]
\centering
\includegraphics[width=1\columnwidth]{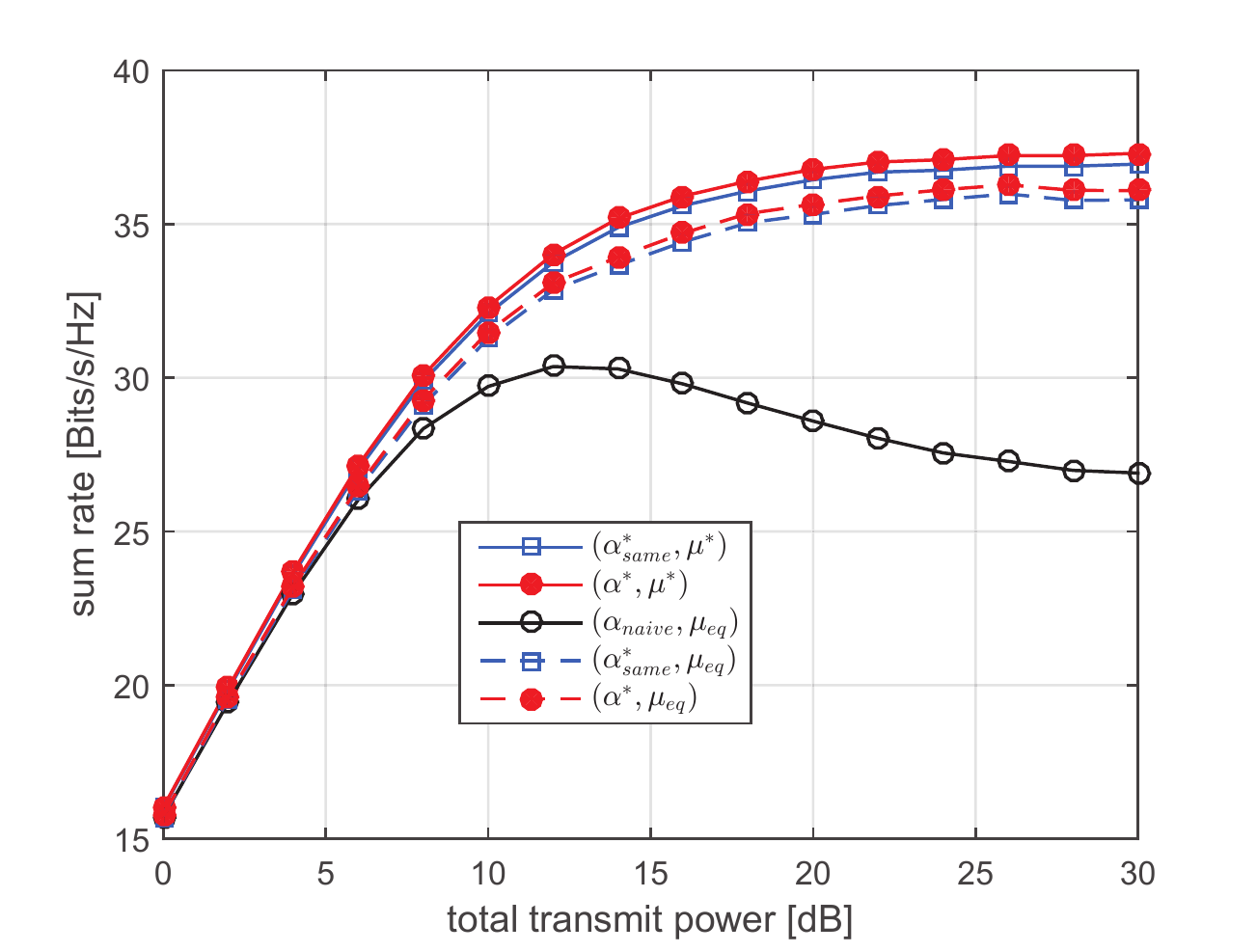}
\caption{Ergodic sum rate as a function of total transmit power, comparison between different transmission algorithms, RZF precoding is implemented.}
\label{jointoptimizepalpha}
\end{figure}

In Fig. \ref{jointoptimizepalpha}, the performance of different transmission algorithms are compared. We can clearly observe the improved robustness and the large performance increase for the proposed algorithm. In this isotropic channel setting, equal power allocation is not a bad strategy since joint optimization $\LB\bm\alpha^{\star}, \bm\mu^{\star}\RB$ only brings a $3\%$ sum rate increase compared to $\LB\bm\alpha^{\star}, \bm\mu_{eq}\RB$. Intriguingly, even if the CSIT accuracy is asymmetric at different TXs, simulation reveals that there is only a negligible performance degradation when imposing identical regularization coefficient at different TXs for isotropic channel setting.
%
%

\section{Conclusion}
In this work, we have studied regularized ZF joint precoding in a distributed CSI configuration. We extend the conventional centralized CSI to distributed CSI scenario by allowing the CSI errors at the different TXs to be arbitrarily correlated. Using RMT tools, an analytical expression is derived to approximate the average rate per user in the large system limit. This deterministic equivalent expression is then used to optimize the regularization coefficients as well as the power allocation at the different TXs in order to reduce the negative impact of the D-CSI configuration.

\FloatBarrier

\appendices

\section{Classical Random Matrix Theory Lemmas} \label{app:literature}

\begin{lemma}[Adapted from \cite{Wagner2012,Couillet2011}]
Let $\alpha^{(j)}>0,j=1,\ldots,n$ and $m_k^{(j)[t]},t \geq 0$ be the sequence defined as
\begin{equation}
\begin{cases}
m_k^{(j)[0]}=\frac{1}{\alpha^{(j)}}& \forall k=1,\ldots,K\\
m_k^{(j)[t]}= \frac{1}{M}\trace\LB \bm{\Theta}_k\LB\frac{1}{M}\sum_{\ell=1}^{K}\frac{\bm{\Theta}_{\ell}}{1+m_{\ell}^{(j)[t-1]}}+\alpha^{(j)}\bI_M\RB^{-1}\RB &\textrm{for $t\geq 1$}
\label{eq:fixed_point}
\end{cases}.
\end{equation}
Then $m_k^{(j)[t]}\xrightarrow{t\rightarrow\infty}m_k^{(j)}$, with $m_k^{(j)}$ solved by constructing an iterative algorithm of \eqref{eq:fixed_point}.
\end{lemma}

\begin{lemma}[Resolvent Identities \cite{Muller2013,Couillet2011}]
\label{lemma_resolvent}
Given any matrix~$\bH\in\mathbb{C}^{K\times M}$, let $\bh^{\He}_k$ denote its $k$th row and $\bH_{[k]}\in\mathbb{C}^{(K-1)\times M}$ denote the matrix obtained after removing the $k$th row from $\bH$. The resolvent matrices of $\bH$ and $\bH_{[k]}$ are denoted by $\bQ= \LB \bH^{\He}\bH+\alpha\I_M\RB^{-1}$ and $\bQ_{[k]}= \LB \bH_{[k]}^{\He}\bH_{[k]}+\alpha\I_M\RB^{-1}$, with $\alpha>0$, respectively. It then holds that
\begin{align}
\bQ&=\bQ_{[k]}-\frac{1}{M}\frac{\bQ_{[k]}\bh_k\bh_k^{\He}\bQ_{[k]}}{1+\frac{1}{M}\bh_k^{\He}\bQ_{[k]}\bh_k}\notag
\end{align}
and
\begin{align}
\bh_k^{\He}\bQ&=\frac{\bh_k^{\He}\bQ_{[k]}}{1+\frac{1}{M}\bh_k^{\He}\bQ_{[k]}\bh_k}.\notag
\end{align}
\end{lemma}

\begin{lemma} [\!\!\cite{Muller2013,Couillet2011}]
\label{lemma_trace}
Let $(\bA_N)_{N\geq 1}, \bA_N\in \mathbb{C}^{N\times N}$ be a sequence of matrices such that $\lim\sup\|\bA_N\|<\infty$, and $(\xv_N)_{N\geq 1}, \xv_N\in \mathbb{C}^{N\times 1}$ be a sequence of random vectors of i.i.d. entries of zero mean, unit variance, and finite eighth order moment independent of~$\bA_N$. Then,
\begin{equation}
\frac{1}{N}\xv_N^{\He}\bA_N\xv_N-\frac{1}{N}\trace\LB \bA_N\RB \xrightarrow[N\rightarrow\infty]{a.s.}0.\notag
\end{equation}
\end{lemma}

\begin{lemma}[\!\!\cite{Muller2013,Couillet2011}]
\label{lemma_zero}
Let $(\bA_N)_{N\geq 1}, \bA_N\in \mathbb{C}^{N\times N}$ be a sequence of matrices such that $\lim\sup\|\bA_N\|<\infty$, and $\xv_N,\yv_N$ be random, mutually independent with i.i.d. entries of zero mean, unit variance, finite eighth order moment, and independent of~$\bA_N$. Then,
\begin{equation}
\frac{1}{N}\xv_N^{\He}\bA_N\yv_N\xrightarrow[N\rightarrow\infty]{a.s.}0.\notag
\end{equation}
\end{lemma}

\begin{lemma}  [\!\!\cite{Wagner2012,Couillet2011}]
\label{lemma_rank1}
Let $\bQ$ and $\bQ_{[k]}$ be as given in Lemma~\ref{lemma_resolvent}. Then, for any matrix $\bA$, we have
\begin{equation}
\trace\LB \bA\LB \bQ-\bQ_{[k]}\RB\RB\leq \|\bA\|_2.\notag
\end{equation}
\end{lemma}

\begin{lemma} [\!\!\cite{Wagner2012,Couillet2011}]
\label{lemma_c_0}
Let $\bU,\bV,\bm{\Theta}$ be of uniformly bounded spectral norm with respect to $N$ and let $\bV$ be invertible. Further, define $\xv={\bm\Theta^{\frac{1}{2}}}\bm{\bz}$ and $\yv= {\bm\Theta^{\frac{1}{2}}}\bm{\bq}$ where $\bz,\bq\in \mathbb{C}^{N}$ have i.i.d. complex entries of zero mean, variance $1/N$	and finite $8$th order moment and be mutually independent as well as independent of~$\bU,\bV$. Define $c_0,c_1,c_2\in \mathbb{R}^{+}$ such that $c_0c_1-c_2^2\geq0$, and let $u= \frac{1}{N}\trace \LB\bm{\Theta}\bV^{-1}\RB$ and $u'= \frac{1}{N}\trace \LB\bm{\Theta}\bU\bV^{-1}\RB$. Then we have:
\begin{align}
&\xv^{\He}\bU \LB \bV+c_0\xv \xv^{\He}+c_1\yv \yv^{\He}+c_2\xv \yv^{\He}+c_2\yv \xv^{\He} \RB^{-1} \xv-\frac{u'\LB 1+c_1 u \RB}{ (c_0c_1-c_2^2)u^2+(c_0+c_1)u+1}\rightarrow 0\notag
\end{align}
as well as
\begin{align}
&\xv^{\He}\bU \LB \bV+c_0\xv \xv^{\He}+c_1\yv \yv^{\He}+c_2\xv \yv^{\He}+c_2\yv \xv^{\He} \RB^{-1} \yv-\frac{-c_2 uu'}{ (c_0c_1-c_2^2)u^2+(c_0+c_1)u+1}\rightarrow 0\notag
\end{align}
\end{lemma}

\section{New Lemmas}


\begin{lemma}
\label{lemma1}
Consider the channel matrices $\hat{\bH}^{(j)},\hat{\bH}^{(j')}$ are distributed according to the D-CSI model in Section \ref{ss:PCCSI}.  Let
\begin{align}
\bQ^{(j)}=\LB\frac{(\hat{\bH}^{(j)})^{\He}\hat{\bH}^{(j)}}{M} +\alpha^{(j)}\I_M\RB^{-1}\notag\\
\bQ^{(j')}=\LB\frac{(\hat{\bH}^{(j')})^{\He}\hat{\bH}^{(j')}}{M} +\alpha^{(j')}\I_M\RB^{-1}\notag
\end{align}
with $\alpha^{(j)},\alpha^{(j')}>0$. Let $\bX \in \mathbb{C}^{M\times M}$ be of uniformly bounded spectral norm with respect to $M$. Then,
\begin{align}
\frac{1}{M^2} \trace \LB\bX \bQ^{(j)}(\hat{\bH}^{(j)})^{\He}\hat{\bH}^{(j')}\bQ^{(j')}\RB-\Gamma_{j,j'}^o(\bX)\xrightarrow{a.s.}0\notag
\end{align}
where the function $\Gamma_{j,j'}^o(\bX):\mathbb{C}^{M\times M}\mapsto \mathbb{C}$ is defined as
\begin{align}
\Gamma_{j,j'}^o(\bX)&=\frac{1}{M}\sum_{k=1}^K \frac{\LB\sqrt{c_{0,k}^{(j)}c_{0,k}^{(j')}}+\sqrt{c_{1,k}^{(j)}c_{1,k}^{(j')}}\rho_{k}^{(j,j')}\RB\frac{1}{M}\trace \LB\bm{\Theta}_{k}\bQ^{(j')}_o\bX \bQ^{(j)}_o\RB }{(1+m^{(j)}_{k})(1+m^{(j')}_{k})}\notag\\
+\frac{1}{M}\sum_{k=1}^K&\frac{ \LB\sqrt{c_{0,k}^{(j)}c_{0,k}^{(j')}}+\sqrt{c_{1,k}^{(j)}c_{1,k}^{(j')}}\rho_{k}^{(j,j')}\RB^2\Gamma^o_{j,j'}(\bm\Theta_{k})\frac{1}{M}\trace \LB\bm{\Theta}_{k}\bQ^{(j')}_o\bX\bQ^{(j)}_o\RB }{(1+m^{(j)}_{k})(1+m^{(j')}_{k})},\notag
\end{align}
with
\begin{equation}
c_{0,k}^{(j)}= 1-(\sigma^{(j)}_k)^2\!, \quad
c_{1,k}^{(j)}= (\sigma^{(j)}_k)^2\!, \quad
c_{2,k}^{(j)}= \sigma^{(j)}_k\sqrt{1\!-\!(\sigma^{(j)}_k)^2}.\notag
\end{equation}
$m_k^{(j)},\bQ^{(j)}_o,m_k^{(j')},\bQ^{(j')}_o$ are defined in Theorem \ref{fundamental_theorem} using $\hat{\bH}^{(j)},\alpha^{(j)},\hat{\bH}^{(j')},\alpha^{(j')}$ respectively.
$\Gamma^o_{j,j'}(\bm\Theta_{k})$ is the $k$th entry of vector $\bm\gamma\in\mathbb{C}^{K\times 1}$. Vector $\bm\gamma$ is the solution for equation system
\begin{align}
\bA
\bm\gamma=\bb.\notag
\end{align}
$\bA\in\mathbb{C}^{K\times K}$ with
\begin{align}
[\bA]_{\ell,t}=1_{\ell=t}-\frac{\LB\sqrt{c_{0,t}^{(j)}c_{0,t}^{(j')}}+\sqrt{c_{1,t}^{(j)}c_{1,t}^{(j')}}\rho_{t}^{(j,j')}\RB^2}{M(1+m^{(j)}_t)(1+m^{(j')}_t)}\frac{\trace\left(\bm\Theta_t\bQ^{(j')}_o\bm\Theta_{\ell}\bQ^{(j)}_o\right)}{M}.\notag
\end{align}
$\bb\in\mathbb{C}^{K\times 1}$ with
\begin{align}
[\bb]_{\ell}=\frac{1}{M}\sum_{k=1}^{K}\frac{\sqrt{c_{0,k}^{(j)}c_{0,k}^{(j')}}+\sqrt{c_{1,k}^{(j)}c_{1,k}^{(j')}}\rho_{k}^{(j,j')}}{(1+m^{(j)}_{k})(1+m^{(j')}_{k})}\frac{\trace\left(\bm\Theta_{k}\bQ^{(j')}_o\bm\Theta_{\ell}\bQ^{(j)}_o\right)}{M}.\notag
\end{align}
\end{lemma}

\begin{proof}
We start by introducing
\begin{align}
\bQ^{(j)}_{[k]}=  \LB \frac{(\hat{\bH}_{[k]}^{(j)})^{\He} \hat{\bH}_{[k]}^{(j)}}{M}+\alpha^{(j)}\I_{M} \RB^{-1}\notag
\end{align}
with
\begin{align}
\hat{\bH}_{[k]}^{(j)}=   \begin{bmatrix}
\hat{\bh}^{(j)}_1 &\hdots&\hat{\bh}^{(j)}_{k-1} &\hat{\bh}^{(j)}_{k+1} &\hdots&\hat{\bh}^{(j)}_K
\end{bmatrix}^{\He}.\notag
\end{align}
$\bQ^{(j')}_{[k]}$ and $\hat{\bH}_{[k]}^{(j')}$ are defined respectively in similar manner as $\bQ^{(j)}_{[k]}$ and $\hat{\bH}_{[k]}^{(j)}$. Let us start by writing the simple equality
\begin{align}
&\bQ^{(j)}-\bQ^{(j)}_o\notag\\
=&\bQ^{(j)}_o \LB (\bQ^{(j)}_o)^{-1}-(\bQ^{(j)})^{-1}\RB \bQ^{(j)}\notag\\
=&\bQ^{(j)}_o \LB \frac{1}{M}\sum_{k=1}^{K}\frac{\bm\Theta_{k}}{1+m_{k}^{(j)}}-\frac{(\hat{\bH}^{(j)})^{\He}\bH^{(j)}}{M} \RB \bQ^{(j)}.
\label{eq:proof_lemma1_3}
\end{align}
We can then replace $\bQ^{(j)}$ using \eqref{eq:proof_lemma1_3} to obtain
\begin{align}
&\frac{1}{M^2} \trace \LB\bX \bQ^{(j)} (\hat{\bH}^{(j)})^{\He}\hat{\bH}^{(j')}\bQ^{(j')}\RB\notag\\
=& \frac{1}{M^2} \trace \LB\bX \bQ^{(j)}_o (\hat{\bH}^{(j)})^{\He}\hat{\bH}^{(j')}\bQ^{(j')}\RB +\sum_{k=1}^K\frac{\trace \LB\bX \bQ^{(j)}_o\bm{\Theta}_{k}\bQ^{(j)}(\hat{\bH}^{(j)})^{\He}\hat{\bH}^{(j')}\bQ^{(j')}\RB}{M^3 \LB 1+m^{(j)}_{k}\RB }\notag\\
&~-\!\frac{1}{M^3} \trace \LB\bX \bQ^{(j)}_o (\hat{\bH}^{(j)})^{\He}\hat{\bH}^{(j)} \bQ^{(j)}(\hat{\bH}^{(j)})^{\He}\hat{\bH}^{(j')}\bQ^{(j')}\RB\notag\\
=& Z_1 +Z_2+Z_3.\notag
\end{align}
We will now calculate separately each of the term~$Z_i$. Starting with~$Z_1$ gives
\begin{align}
Z_1&= \frac{1}{M^2} \trace \LB\bX \bQ^{(j)}_o (\hat{\bH}^{(j)})^{\He}\hat{\bH}^{(j')}\bQ^{(j')}\RB\notag\notag\\
&= \frac{1}{M} \sum_{k=1}^K\frac{1}{M}(\hat{\bh}_{k}^{(j')})^{\He} \bQ^{(j')}\bX \bQ^{(j)}_o \hat{\bh}^{(j)}_{k}\notag\\
&\stackrel{(a)}{=} \frac{1}{M} \sum_{k=1}^K \frac{1}{M}\frac{(\hat{\bh}_{k}^{(j')})^{\He}\bQ_{[k]}^{(j')}\bX \bQ^{(j)}_o \hat{\bh}_{k}^{(j)}}{1+\frac{1}{M}(\hat{\bh}_{k}^{(j')})^{\He} \bQ_{[k]}^{(j')}\hat{\bh}_{k}^{(j')}}\notag\\
&\stackrel{(b)}{\asymp} \frac{1}{M} \sum_{k=1}^K \frac{\LB\sqrt{c_{0,k}^{(j)}c_{0,k}^{(j')}}+\sqrt{c_{1,k}^{(j)}c_{1,k}^{(j')}}\rho_{k}^{(j,j')}\RB\frac{1}{M}\trace \LB \bm{\Theta}_{k}\bQ_{[k]}^{(j')}\bX \bQ^{(j)}_o\RB}{1+\frac{1}{M}\trace \LB\bm\Theta_k \bQ_{[k]}^{(j')} \RB}\notag\\
&\stackrel{(c)}{\asymp} \frac{1}{M} \sum_{k=1}^K \frac{\LB\sqrt{c_{0,k}^{(j)}c_{0,k}^{(j')}}+\sqrt{c_{1,k}^{(j)}c_{1,k}^{(j')}}\rho_{k}^{(j,j')}\RB\frac{1}{M}\trace \LB \bm{\Theta}_{k}\bQ^{(j')}\bX \bQ^{(j)}_o \RB}{1+\frac{1}{M}\trace \LB \bm\Theta_k\bQ^{(j')} \RB}\notag\\
&\stackrel{(d)}{\asymp} \frac{1}{M} \sum_{k=1}^K\frac{ \LB\sqrt{c_{0,k}^{(j)}c_{0,k}^{(j')}}+\sqrt{c_{1,k}^{(j)}c_{1,k}^{(j')}}\rho_{k}^{(j,j')}\RB\frac{1}{M}\trace \LB\bm{\Theta}_{k}\bQ^{(j')}_0\bX\bQ^{(j)}_0\RB}{1+m^{(j')}_{k}},\notag
\end{align}
where equality~$(a)$ follows from Lemma~\ref{lemma_resolvent}, equality~$(b)$ from Lemma~\ref{lemma_trace}, equality~$(c)$ from Lemma~\ref{lemma_rank1}, and equality~$(d)$ from the fundamental Theorem \ref{fundamental_theorem}. The following calculations are very similar and the same lemmas are used.

Turning to $Z_3$ gives
\begin{align}
Z_3&=-\frac{1}{M^3} \trace \LB\bX \bQ^{(j)}_o (\hat{\bH}^{(j)})^{\He}\hat{\bH}^{(j)} \bQ^{(j)}(\hat{\bH}^{(j)})^{\He}\hat{\bH}^{(j')}\bQ^{(j')}\RB\notag\\
&=-\frac{1}{M^3} \sum_{k=1}^K \trace \LB(\hat{\bh}_{k}^{(j)})^{\He} \bQ^{(j)}(\hat{\bH}^{(j)})^{\He}\hat{\bH}^{(j')}\bQ^{(j')}\bX \bQ^{(j)}_o \hat{\bh}_{k}^{(j)}\RB\notag\\
&=-\frac{1}{M^3} \sum_{k=1}^K \frac{\trace \LB(\hat{\bh}_{k}^{(j)})^{\He} \bQ_{[k]}^{(j)}(\hat{\bH}^{(j)})^{\He}\hat{\bH}^{(j')}\bQ^{(j')}\bX \bQ^{(j)}_o \hat{\bh}_{k}^{(j)}\RB}{1+\frac{1}{M}(\hat{\bh}_{k}^{(j)})^{\He} \bQ_{[k]}^{(j)}\hat{\bh}_{k}^{(j)}}\notag\\
&\stackrel{(e)}{=}-\frac{1}{M^3} \sum_{k=1}^K \frac{\trace \LB(\hat{\bh}_{k}^{(j)})^{\He} \bQ_{[k]}^{(j)}(\hat{\bH}^{(j)})^{\He}\hat{\bH}^{(j')}\bQ_{[k]}^{(j')}\bX \bQ^{(j)}_o \hat{\bh}_{k}^{(j)}\RB}{1+\frac{1}{M}(\hat{\bh}_{k}^{(j)})^{\He} \bQ_{[k]}^{(j)}\hat{\bh}_{k}^{(j)}}\notag\\
&~+\frac{1}{M^4} \sum_{k=1}^K \frac{\trace \LB(\hat{\bh}_{k}^{(j)})^{\He} \bQ_{[k]}^{(j)}(\hat{\bH}^{(j)})^{\He}\hat{\bH}^{(j')}\bQ_{[k]}^{(j')}\hat{\bh}_{k}^{(j')}(\hat{\bh}_{k}^{(j')})^{\He}\bQ_{[k]}^{(j')}\bX \bQ^{(j)}_o \hat{\bh}_{k}^{(j)}\RB}{\LB 1+\frac{1}{M}(\hat{\bh}_{k}^{(j)})^{\He} \bQ_{[k]}^{(j)}\hat{\bh}_{k}^{(j)}\RB\LB 1+\frac{1}{M}(\hat{\bh}_{k}^{(j')})^{\He} \bQ_{[k]}^{(j')}\hat{\bh}_{k}^{(j')}\RB}\notag\\
&= Z_4+Z_5,\notag
\end{align}
with equality~$(e)$ obtained using Lemma~\ref{lemma_resolvent} for $\bQ^{(j')}$. We also split the calculation in two and start by calculating $Z_4$ as follows.
\begin{align}
Z_4&=-\frac{1}{M^3} \sum_{k=1}^K \frac{\trace \LB(\hat{\bh}_{k}^{(j)})^{\He} \bQ_{[k]}^{(j)}(\hat{\bH}_{[k]}^{(j)})^{\He}\hat{\bH}_{[k]}^{(j')}\bQ_{[k]}^{(j')}\bX \bQ^{(j)}_o \hat{\bh}_{k}^{(j)}\RB}{1+\frac{1}{M}(\hat{\bh}_{k}^{(j)})^{\He} \bQ_{[k]}^{(j)}\hat{\bh}_{k}^{(j)}}\notag\\
&~-\frac{1}{M^3} \sum_{k=1}^K \frac{\trace \LB(\hat{\bh}_{k}^{(j)})^{\He} \bQ_{[k]}^{(j)}\hat{\bh}_{k}^{(j)}(\hat{\bh}_{k}^{(j')})^{\He}\bQ_{[k]}^{(j')}\bX \bQ^{(j)}_o \hat{\bh}_{k}^{(j)}\RB}{1+\frac{1}{M}(\hat{\bh}_{k}^{(j)})^{\He} \bQ_{[k]}^{(j)}\hat{\bh}_{k}^{(j)}}\notag\\
&\asymp -\frac{1}{M^3} \sum_{k=1}^K \frac{\trace \LB\bm{\Theta}_{k}\bQ_{[k]}^{(j)}(\hat{\bH}_{[k]}^{(j)})^{\He}\hat{\bH}_{[k]}^{(j')}\bQ_{[k]}^{(j')}\bX \bQ^{(j)}_o\RB}{1+\frac{1}{M}\trace \LB\bm{\Theta}_{k}\bQ_{[k]}^{(j)}\RB}\notag\\
&~~~~-\frac{1}{M} \sum_{k=1}^K \LB\sqrt{c_{0,k}^{(j)}c_{0,k}^{(j')}}+\sqrt{c_{1,k}^{(j)}c_{1,k}^{(j')}}\rho_{k}^{(j,j')}\RB\cdot\frac{\frac{1}{M}\trace \LB \bm{\Theta}_{k}\bQ_{[k]}^{(j)}\RB \frac{1}{M} \trace \LB\bm{\Theta}_{k}\bQ_{[k]}^{(j')}\bX \bQ^{(j)}_o\RB}{1+\frac{1}{M}\trace \LB\bm{\Theta}_{k}\bQ_{[k]}'\RB}\notag\\
&\stackrel{(f)}{\asymp} -\frac{1}{M} \sum_{k=1}^K \frac{\frac{1}{M^2}\trace \LB\bm{\Theta}_{k}\bQ^{(j)} (\hat{\bH}^{(j)})^{\He}\hat{\bH}^{(j')}\bQ^{(j')}\bX \bQ^{(j)}_o\RB}{1+m^{(j)}_{k}}\notag\\
&~-\frac{1}{M} \sum_{k=1}^K \LB\sqrt{c_{0,k}^{(j)}c_{0,k}^{(j')}}+\sqrt{c_{1,k}^{(j)}c_{1,k}^{(j')}}\rho_{k}^{(j,j')}\RB\frac{m^{(j)}_{k}\frac{1}{M} \trace \LB\bm{\Theta}_{k}\bQ^{(j')}_0\bX \bQ^{(j)}_o\RB}{1+m^{(j)}_{k}}\notag\\
&\asymp -Z_2-\frac{1}{M} \sum_{k=1}^K \frac{\LB\sqrt{c_{0,k}^{(j)}c_{0,k}^{(j')}}+\sqrt{c_{1,k}^{(j)}c_{1,k}^{(j')}}\rho_{k}^{(j,j')}\RB\frac{m^{(j)}_{k}}{M} \trace \LB\bm{\Theta}_{k}\bQ^{(j')}_0\bX \bQ^{(j)}_o\RB}{1+m^{(j)}_{k}},\notag
\end{align}
where $(f)$ applies multiple times Lemma \ref{lemma_rank1}.

Finally, $Z_5$ is calculated as
\begin{align}
Z_5&\asymp \frac{1}{M^4} \sum_{k=1}^K \frac{\trace \LB(\hat{\bh}_{k}^{(j)})^{\He} \bQ_{[k]}^{(j)}(\hat{\bH}^{(j)})^{\He}\hat{\bH}^{(j')}\bQ_{[k]}^{(j')}\hat{\bh}_{k}^{(j')}(\hat{\bh}_{k}^{(j')})^{\He}\bQ_{[k]}^{(j')}\bX \bQ^{(j)}_o \hat{\bh}_{k}^{(j)}\RB}{\LB 1+ m^{(j)}_{k}\RB\LB 1+ m^{(j')}_{k}\RB}\notag\\
&= \frac{1}{M^4} \sum_{k=1}^K \frac{\trace \LB(\hat{\bh}_{k}^{(j)})^{\He} \bQ_{[k]}^{(j)}(\hat{\bH}_{[k]}^{(j)})^{\He}\hat{\bH}_{[k]}^{(j')}\bQ_{[k]}^{(j')}\hat{\bh}_{k}^{(j')}(\hat{\bh}_{k}^{(j')})^{\He}\bQ_{[k]}^{(j')}\bX \bQ^{(j)}_o \hat{\bh}_{k}^{(j)}\RB}{\LB 1+ m^{(j)}_{k}\RB\LB 1+ m^{(j')}_{k}\RB}\notag\\
&~+\frac{1}{M^4} \sum_{k=1}^K \frac{\trace \LB(\hat{\bh}_{k}^{(j)})^{\He} \bQ_{[k]}^{(j)}\hat{\bh}_{k}^{(j)}(\hat{\bh}_{k}^{(j')})^{\He}\bQ_{[k]}^{(j')}\hat{\bh}_{k}^{(j')}(\hat{\bh}_{k}^{(j')})^{\He}\bQ_{[k]}^{(j')}\bX \bQ^{(j)}_o \hat{\bh}_{k}^{(j)}\RB}{\LB 1+ m^{(j)}_{k}\RB\LB 1+ m^{(j')}_{k}\RB}\notag\\
&\asymp \frac{1}{M} \sum_{k=1}^K \LB\sqrt{c_{0,k}^{(j)}c_{0,k}^{(j')}}+\sqrt{c_{1,k}^{(j)}c_{1,k}^{(j')}}\rho_{k}^{(j,j')}\RB^2\notag\\
&~~~~~~~~~~~~~\cdot\frac{\frac{1}{M^2}\trace \LB \bm{\Theta}_{k}\bQ_{[k]}^{(j)}(\hat{\bH}_{[k]}^{(j)})^{\He}\hat{\bH}_{[k]}^{(j')}\bQ_{[k]}^{(j')}\RB \frac{1}{M} \trace \LB \bm{\Theta}_{k}\bQ_{[k]}^{(j')}\bX \bQ^{(j)}_o\RB}{\LB 1+ m^{(j)}_{k}\RB\LB 1+ m^{(j')}_{k}\RB}\notag\\
&~+\frac{1}{M} \sum_{k=1}^K\LB\sqrt{c_{0,k}^{(j)}c_{0,k}^{(j')}}+\sqrt{c_{1,k}^{(j)}c_{1,k}^{(j')}}\rho_{k}^{(j,j')}\RB\notag\\
&~~~~~~~~~~~~~\cdot\frac{\frac{1}{M}	 \trace \LB \bm{\Theta}_{k}\bQ_{[k]}^{(j)}\RB \frac{1}{M}	 \trace \LB \bm{\Theta}_{k}\bQ_{[k]}^{(j')}\RB \frac{1}{M} \trace \LB \bm{\Theta}_{k}\bQ_{[k]}^{(j')}\bX \bQ^{(j)}_o\RB}{\LB 1+ m^{(j)}_{k}\RB\LB 1+ m^{(j')}_{k}\RB}\notag\\
&\asymp \frac{1}{M} \sum_{k=1}^K \LB\sqrt{c_{0,k}^{(j)}c_{0,k}^{(j')}}+\sqrt{c_{1,k}^{(j)}c_{1,k}^{(j')}}\rho_{k}^{(j,j')}\RB^2\notag\\
&~~~~~~~~~~~~~\cdot\frac{\frac{1}{M^2}\trace \LB \bm{\Theta}_{k}\bQ^{(j)}(\hat{\bH}^{(j)})^{\He}\hat{\bH}^{(j')}\bQ^{(j')}\RB \frac{1}{M} \trace \LB \bm{\Theta}_{k}\bQ^{(j')}_o\bX \bQ^{(j)}_o\RB}{\LB 1+ m^{(j)}_{k}\RB\LB 1+ m^{(j')}_{k}\RB}\notag\\
&~+\frac{1}{M} \sum_{k=1}^K \LB\sqrt{c_{0,k}^{(j)}c_{0,k}^{(j')}}+\sqrt{c_{1,k}^{(j)}c_{1,k}^{(j')}}\rho_{k}^{(j,j')}\RB\frac{  m^{(j)}_{k}m^{(j')}_{k} \frac{\trace \LB \bm{\Theta}_{k}\bQ^{(j)}_o\bX \bQ^{(j)}_o\RB}{M}}{\LB 1+ m^{(j)}_{k}\RB\LB 1+ m^{(j')}_{k}\RB}.\notag
\end{align}
Adding all the $Z_i$ gives
\begin{align}
&\frac{1}{M^2} \trace \LB\bX \bQ^{(j)} (\hat{\bH}^{(j)})^{\He}\hat{\bH}^{(j')}\bQ^{(j')}\RB\notag\\
&\asymp \frac{1}{M}\sum_{k=1}^K \frac{\LB\sqrt{c_{0,k}^{(j)}c_{0,k}^{(j')}}+\sqrt{c_{1,k}^{(j)}c_{1,k}^{(j')}}\rho_{k}^{(j,j')}\RB\frac{1}{M}\trace \LB\bm{\Theta}_{k}\bQ^{(j')}_o\bX \bQ^{(j)}_o\RB }{(1+m^{(j)}_{k})(1+m^{(j')}_{k})}\notag\\
&~+\frac{1}{M}\sum_{k=1}^K\frac{1}{M^2} \trace \LB \bm{\Theta}_{k}\bQ^{(j)} (\hat{\bH}^{(j)})^{\He}\hat{\bH}^{(j')}\bQ^{(j')}\RB\notag\\
&~~~~~~~~~~~\cdot\frac{\LB\sqrt{c_{0,k}^{(j)}c_{0,k}^{(j')}}+\sqrt{c_{1,k}^{(j)}c_{1,k}^{(j')}}\rho_{k}^{(j,j')}\RB^2\frac{1}{M}\trace \LB\bm{\Theta}_{k}\bQ^{(j')}_o\bX \bQ^{(j)}_o\RB }{(1+m^{(j)}_{k})(1+m^{(j')}_{k})},\notag
\end{align}
which finally gives
\begin{align}\label{eq:proof_lemma1_9}
\Gamma^o_{j,j'}(\bX)&=\frac{1}{M}\sum_{k=1}^K \frac{\LB\sqrt{c_{0,k}^{(j)}c_{0,k}^{(j')}}+\sqrt{c_{1,k}^{(j)}c_{1,k}^{(j')}}\rho_{k}^{(j,j')}\RB\frac{1}{M}\trace \LB\bm{\Theta}_{k}\bQ^{(j')}_o\bX \bQ^{(j)}_o\RB }{(1+m^{(j)}_{k})(1+m^{(j')}_{k})}\notag\\
+\frac{1}{M}\sum_{k=1}^K&\frac{ \LB\sqrt{c_{0,k}^{(j)}c_{0,k}^{(j')}}+\sqrt{c_{1,k}^{(j)}c_{1,k}^{(j')}}\rho_{k}^{(j,j')}\RB^2\Gamma^o_{j,j'}(\bm\Theta_{k})\frac{1}{M}\trace \LB\bm{\Theta}_{k}\bQ^{(j')}_o\bX\bQ^{(j)}_o\RB }{(1+m^{(j)}_{k})(1+m^{(j')}_{k})}.
\end{align}
It remains then to calculate $\Gamma^o_{j,j'}(\bm\Theta_{k})$ to conclude the calculation. Indeed, it is the solution of equation system when asserting $\bX=\bm{\Theta}_{k},\forall k=1,\ldots,K$ into \eqref{eq:proof_lemma1_9}.
\end{proof}

\begin{lemma}
\label{lemma2}
Let $\bL,\bR,\bar{\bA},\bm\Theta\in \mathbb{C}^{M \times M}$ be of uniformly bounded spectral norm with respect to $M$ and let $\bar{\bA}$ be invertible. Further define $\xv=\bm\Theta^{\frac{1}{2}}\zv$, $\xv'=\bm\Theta^{\frac{1}{2}}\zv'$ and $\yv=\bm\Theta^{\frac{1}{2}}\qv$. $\zv,\zv'$ satisfies $\zv=\rho\zv'+\sqrt{1-\rho^2}\wv$. $\zv,\qv$ and $\zv',\qv,\wv$ are mutually independent as well as independent of $\bL,\bR,\bar{\bA}$. $\zv,\zv',\qv,\wv$ have i.i.d. complex entries of zero mean, variance $1/M$ and finite $8$th order moment. Let us define
\begin{align*}
\bA&=\bar{\bA}+c_0\mathbf{x}\mathbf{x}^{\He}+c_1\yv\yv^{\He}+c_2\mathbf{x}\yv^{\He}+c_2\yv\mathbf{x}^{\He}\\
\bA'&=\bar{\bA}+c_0\mathbf{x}'\mathbf{x}'^{\He}+c_1\yv\yv^{\He}+c_2\mathbf{x}'\yv^{\He}+c_2\yv\mathbf{x}'^{\He},
\end{align*}
let $c_0,c_1,c_2\in\mathbb{R}^{+}$ with~$c_0+c_1=1$ and $c_0c_1-c_2^2=0$, and
\begin{align*}
u &=\frac{\trace( \bm\Theta\bar{\bA}^{-1})}{M},\qquad \;\;\;\;\; u_{\mathrm{L}}=\frac{\trace(\bm\Theta\bL \bar{\bA}^{-1})}{M},\\
u_{\mathrm{R}}&=\frac{\trace(\bm\Theta\bar{\bA}^{-1}\bR )}{M},\qquad u_{\mathrm{LR}}=\frac{\trace(\bm\Theta\bL\bar{\bA}^{-1}\bR )}{M}.
\end{align*}

Then we have:
\begin{align*}
\xv^{\He}\bL\bA^{-1}\bR\xv &\asymp u_{\mathrm{LR}}-\frac{c_0u_{\mathrm{L}}u_{\mathrm{R}}}{1+u}\\
\xv^{\He}\bL\bA^{-1}\bR\yv &\asymp -\frac{c_2 u_{\mathrm{L}}u_{\mathrm{R}}}{1+u}\\
\xv^{\He}\bL\bA'^{-1}\bR\yv &\asymp -\rho\frac{c_2 u_{\mathrm{L}}u_{\mathrm{R}}}{1+u}.
\end{align*}
\end{lemma}

\begin{proof}
Focusing first on the first equality gives
\begin{align*}
&\xv^{\He}\bL\bA^{-1}\bR\xv-\xv^{\He}\bL\bar{\bA}^{-1}\bR\xv\\
&=\xv^{\He}\bL\bA^{-1}\LB\bar{\bA}-\bA \RB   \bar{\bA}^{-1}\bR\xv \\
&=-\xv^{\He}\bL\bA^{-1}\LB c_0 \xv \xv^{\He}+c_1\yv \yv^{\He}+c_2 \yv \xv^{\He}+c_2\xv \yv^{\He} \RB   \bar{\bA}^{-1}\bR\xv \\
&\stackrel{(a)}{\asymp}- \LB  c_0\xv^{\He}\bL\bA^{-1}\xv +c_2 \xv^{\He}\bL\bA^{-1}\yv \RB\trace \LB \bm\Theta\bar{\bA}^{-1}\bR\RB \\
&\stackrel{(b)}{\asymp}- c_0 \frac{\trace \LB \bm\Theta\bL\bar{\bA}^{-1}\RB}{M} \frac{\trace \LB \bm\Theta\bar{\bA}^{-1}\bR\RB}{M} \frac{1+c_1\frac{\trace \LB \bm\Theta\bar{\bA}^{-1}\RB}{M}}{1+\frac{\trace \LB \bm\Theta\bar{\bA}^{-1}\RB}{M}}\notag\\
&~~~~+c_2^2 \frac{\trace \LB \bm\Theta\bL\bar{\bA}^{-1}\RB}{M} \frac{\trace \LB \bm\Theta\bar{\bA}^{-1}\bR\RB}{M}\frac{\frac{\trace \LB \bm\Theta\bar{\bA}^{-1}\RB}{M}}{1+\frac{\trace \LB \bm\Theta\bar{\bA}^{-1}\RB}{M}},
\end{align*}
where equality $(a)$ is obtained from using Lemma~\ref{lemma_zero} and Lemma~\ref{lemma_trace} and equality~$(b)$ follows from Lemma~\ref{lemma_c_0}.
Similarly, we turn to the second equality to write
\begin{align*}
&\xv^{\He}\bL\bA^{-1}\bR\yv-\xv^{\He}\bL\bar{\bA}^{-1}\bR\yv\\
=&\xv^{\He}\bL\bA^{-1}\LB\bar{\bA}-\bA \RB   \bar{\bA}^{-1}\bR\yv \\
=&-\xv^{\He}\bL\bA^{-1}\LB c_0 \xv \xv^{\He}+c_1\yv \yv^{\He}+c_2 \yv \xv^{\He}+c_2\xv \yv^{\He} \RB   \bar{\bA}^{-1}\bR\yv \\
\stackrel{(c)}{\asymp}&-\LB c_1\xv^{\He}\bL\bA^{-1}\yv  +c_2\xv^{\He}\bL\bA^{-1}\xv   \RB   \frac{\trace\LB \bm\Theta\bar{\bA}^{-1}\bR\RB}{M} \\
\stackrel{(d)}{\asymp}& c_1  c_2 \frac{\trace \LB \bm\Theta\bL\bar{\bA}^{-1}\RB}{M} \frac{\trace \LB \bm\Theta\bar{\bA}^{-1}\bR\RB}{M} \frac{\frac{\trace \LB \bm\Theta\bar{\bA}^{-1}\RB}{M}}{1+\frac{\trace \LB \bm\Theta\bar{\bA}^{-1}\RB}{M}}\notag\\
&-c_2\frac{\trace \LB \bm\Theta\bL\bar{\bA}^{-1}\RB}{M} \frac{\trace \LB \bm\Theta\bar{\bA}^{-1}\bR\RB}{M}\frac{1+c_1\frac{\trace \LB \bm\Theta\bar{\bA}^{-1}\RB}{M}}{1+\frac{\trace \LB \bm\Theta\bar{\bA}^{-1}\RB}{M}},
\end{align*}
where equality $(c)$ is obtained from using Lemma~\ref{lemma_zero} and Lemma~\ref{lemma_trace} and equality~$(d)$ follows from Lemma~\ref{lemma_c_0}.
For the third equality,
\begin{align*}
&\xv^{\He}\bL\bA'^{-1}\bR\yv\\
=& \rho\xv'\bL\bA'^{-1}\bR\yv+\sqrt{1-\rho^2}\bm\Theta^{\frac{1}{2}}\wv\bL\bA'^{-1}\bR\yv\\
\stackrel{(e)}{\asymp}& \rho\xv'\bL\bA'^{-1}\bR\yv\\
\asymp &-\rho\frac{c_2 u_{\mathrm{L}}u_{\mathrm{R}}}{1+u},
\end{align*}
where equality $(e)$ is obtained from using Lemma~\ref{lemma_zero}.
\end{proof}


\section{Proof of Deterministic Equivalent Theorem~\ref{theorem}}\label{se:proof}

The proof is built upon results from both \cite{Wagner2012} and \cite{Muller2013} and novel lemmas Lemma~\ref{lemma1} and Lemma~\ref{lemma2}. We also make extensive use of classical RMT lemmas recalled in Appendix~\ref{app:literature}. In particular, Lemma~\ref{lemma1} extends \cite[Lemma~$15$]{Muller2013} and is an interesting result in itself.

\subsection[Deterministic equivalent for power normalization term]{Deterministic equivalent for $\Psi^{(j)}$}\label{se:proof:preliminaries1}

We start by finding a deterministic equivalent for $\Psi^{(j)}$. Apply Lemma~\ref{lemma1} with $\hat{\bH}^{(j')}=\hat{\bH}^{(j)},~\bA=\bI_M$, which gives
\begin{align}
\Psi^{(j)}&\asymp\Gamma^o_{j,j}(\I_M)\notag\\
&=\LB\frac{1}{M}\sum_{\ell=1}^K \frac{\frac{1}{M}\trace \LB\bm{\Theta}_{\ell}\bQ_o^{(j)}\bQ_o^{(j)}\RB }{(1+m^{(j)}_{\ell})^2}+\frac{1}{M}\sum_{\ell=1}^K\frac{ \Gamma^o_{j,j}(\bm\Theta_{\ell})\frac{1}{M}\trace \LB\bm{\Theta}_{\ell}\bQ_o^{(j)}\bQ_o^{(j)}\RB }{(1+m^{(j)}_{\ell})^2}\RB.
\label{eq:main_proof_1}
\end{align}
From \eqref{eq:main_proof_1}, it can be noted that, as expected, this deterministic equivalent does not depend on~$\sigma^{(j)}_{\ell}$.
The total power constraint for large scale system reads
\begin{align*}
&\parallel\bT_{\rZF}^{\DCSI}\parallel_{\Fro}^2 \\
&=\sum_{j=1}^n \mu_j^2\tr\LB\bE_j^{\He}\bT_{\rZF}^{(j)}(\bT_{\rZF}^{(j)})^{\He}\bE_j\RB\\
&\stackrel{(a)}{\asymp}\sum_{j=1}^n\mu_j^2\frac{P}{\Gamma_{j,j}^{o}(\bI_M)}\Gamma_{j,j}^{o}(\bE_j\bE_j^{\He})\\
&=P,
\end{align*}
where $(a)$ follows from Lemma \ref{lemma1}. Therefore, there is a constraint for the power scaling factors $\mu_j$:
\begin{align*}
\sum_{j=1}^n\mu_j^2\frac{\Gamma_{j,j}^{o}(\bE_j\bE_j^{\He})}{\Gamma_{j,j}^{o}(\bI_M)}=1.
\end{align*}

\subsection[Deterministic equivalent for individual interference term]{Deterministic equivalent for $\bh_k^{\He}\bt^{\DCSI}_{\rZF,k}$}\label{se:proof:preliminaries2}
Turning to the desired signal at RX~$k$, we can write
\begin{align*}
\bh_k^{\He}\bt_{\rZF,k}^{\DCSI}
&=  \sum_{j=1}^n \frac{1}{M}\frac{\mu_j\sqrt{P}}{\sqrt{\Psi^{(j)}}}\bh_k^{\He}\bE_j\bE_j^{\He}(\bC^{(j)})^{-1}\hat{\bh}^{(j)}_k\\
&\stackrel{(a)}{\asymp} \sqrt{P}\sum_{j=1}^n \mu_j\sqrt{\frac{1}{\Gamma^o_{j,j}(\bI_M)}}\frac{\frac{1}{M}\bh_k^{\He}\bE_j\bE_j^{\He}(\bC_{[k]}^{(j)})^{-1}\hat{\bh}^{(j)}_k}{1+\frac{1}{M}(\hat{\bh}^{(j)}_k)^{\He} (\bC_{[k]}^{(j)})^{-1} \hat{\bh}^{(j)}_k}\\
&\stackrel{(b)}{\asymp}  \sqrt{P}\sum_{j=1}^n \mu_j\sqrt{\frac{1-(\sigma^{(j)}_k)^2}{\Gamma^o_{j,j}(\bI_M)}}  \frac{\frac{1}{M}\bh_k^{\He}\bE_j\bE_j^{\He}(\bC_{[k]}^{(j)})^{-1}\bh_k}{1+\frac{1}{M}(\hat{\bh}^{(j)}_k)^{\He} (\bC_{[k]}^{(j)})^{-1} \hat{\bh}^{(j)}_k}\\
& \stackrel{(c)}{\asymp}\sqrt{P} \sum_{j=1}^n  \mu_j\sqrt{\frac{1-(\sigma^{(j)}_k)^2}{\Gamma^o_{j,j}(\bI_M)}} \frac{ \frac{1}{M}\trace \LB \bm{\Theta}_k\bE_j \bE_j^{\He}(\bC_{[k]}^{(j)})^{-1}\RB}{1+\frac{1}{M} \trace \LB \bm{\Theta}_k (\bC_{[k]}^{(j)})^{-1} \RB} \\	
& \stackrel{(d)}{\asymp}\sqrt{P} \sum_{j=1}^n  \mu_j\sqrt{\frac{1-(\sigma^{(j)}_k)^2}{\Gamma^o_{j,j}(\bI_M)}} \frac{ \frac{1}{M}\trace \LB \bm{\Theta}_k\bE_j \bE_j^{\He}(\bC^{(j)})^{-1}\RB}{1+\frac{1}{M} \trace \LB \bm{\Theta}_k (\bC^{(j)})^{-1} \RB} \\
& \stackrel{(e)}\asymp\sqrt{P} \sum_{j=1}^n  \mu_j\sqrt{\frac{1-(\sigma^{(j)}_k)^2}{\Gamma^o_{j,j}(\bI_M)}} \frac{ \frac{1}{M}\trace \LB \bm{\Theta}_k\bE_j \bE_j^{\He}\bQ^{(j)}_o\RB}{1+\frac{1}{M} \trace \LB \bm{\Theta}_k \bQ^{(j)}_o \RB},
\end{align*}
where we have defined
\begin{align*}
\bC_{[k]}^{(j)}=  \frac{ \hat{\bH}_{[k]}^{(j)}(\hat{\bH}_{[k]}^{(j)})^{\He}}{M}+\alpha^{(j)}\I_{M},\qquad \forall j
\end{align*}
with
\begin{align*}
(\hat{\bH}_{[k]}^{(j)})^{\He}=   \begin{bmatrix}
\hat{\bh}_1^{(j)} &\hdots&\hat{\bh}_{k-1}^{(j)} &\hat{\bh}_{k+1}^{(j)} &\hdots&\hat{\bh}_K^{(j)}
\end{bmatrix},\qquad \forall j.
\end{align*}

Equality $(a)$ follows then from Lemma~\ref{lemma_resolvent} and the use of the deterministic equivalent derived for~$\Psi^{(j)}$, $(b)$ from Lemma~\ref{lemma_zero}, $(c)$ from Lemma~\ref{lemma_trace}, $(d)$ from Lemma~\ref{lemma_rank1} and $(e)$ from the fundamental Theorem~\ref{fundamental_theorem}.

It follows then directly that
\begin{align*}
\left |\bh_k^{\He}\bt_{\rZF,k}^{\DCSI}\right |^2&\asymp P\LB \sum_{j=1}^n  \mu_j\sqrt{\frac{1-(\sigma^{(j)}_k)^2}{\Gamma^o_{j,j}(\bI_M)}} \frac{ \frac{1}{M}\trace \LB \bm{\Theta}_k\bE_j \bE_j^{\He}\bQ^{(j)}_o\RB}{1+\frac{1}{M} \trace \LB \bm{\Theta}_k \bQ^{(j)}_o \RB}\RB^2.
\end{align*}

\subsection[Deterministic equivalent for the interference term]{Deterministic Equivalent for $\mathcal{I}_k$}\label{se:proof:preliminaries}
Our first step is to write explicitly the interference term using the definition of~$\bT^{\DCSI}$ and replace~$\Psi^{(j)}$ by its deterministic equivalent.
\begin{align}
\mathcal{I}_k&= \sum_{\ell=1,\ell\neq k}^K|\bh_k^{\He}\bt_{\rZF,\ell}^{\DCSI}|^2\notag\\
&=\bh_k^{\He}\bT_{\rZF}^{\DCSI}  (\bT_{\rZF}^{\DCSI})^{\He}\bh_k-\bh_k^{\He}\bt_{\rZF,k}^{\DCSI}(\bt_{\rZF,k}^{\DCSI})^{\He}\bh_k\notag\\
&=\!\frac{1}{M^2}\!\sum_{j=1}^n\!\sum_{j'=1}^n\!\frac{\mu_j\mu_{j'}P}{\sqrt{\Psi^{(j)}\Psi^{(j')}}}\bh_k^{\He} \bE_j\!\bE_j^{\He}(\bC^{(j)})^{-1}(\hat{\bH}_{[k]}^{(j)})^{\He} \hat{\bH}_{[k]}^{(j')}(\bC^{(j')})^{-1}\bE_{j'}\bE_{j'}^{\He} \bh_k\notag\\
&\asymp \frac{P}{M^2} \sum_{j=1}^n\sum_{j'=1}^n\frac{\mu_j\mu_{j'}}{\sqrt{\Gamma^o_{j,j}(\bI_M)\Gamma^o_{j',j'}(\bI_M)}} \bh_k^{\He}\bE_j\bE_j^{\He}(\bC_{[k]}^{(j)})^{-1}\notag\\
&~~~~~~~~~~~~~~~~~~\cdot(\hat{\bH}_{[k]}^{(j)})^{\He}\hat{\bH}_{[k]}^{(j')}(\bC^{({j'})})^{-1} \bE_{j'}\bE_{j'}^{\He}\bh_k\notag\\
&+\frac{P}{M^2} \sum_{j=1}^n\sum_{j'=1}^n \frac{\mu_j\mu_{j'}}{\sqrt{\Gamma^o_{j,j}(\bI_M)\Gamma^o_{j',j'}(\bI_M)}}\bh_k^{\He} \bE_j\bE_j^{\He}\LB(\bC^{(j)})^{-1}-(\bC_{[k]}^{(j)})^{-1}\RB\notag\\
&~~~~~~~~~~~~~~~~~~\cdot(\hat{\bH}_{[k]}^{(j)})^{\He}\hat{\bH}_{[k]}^{(j')}(\bC^{({j'})})^{-1}\bE_{j'}\bE_{j'}^{\He} \bh_k.
\label{eq:main_proof_6}
\end{align}
To obtain a deterministic equivalent for the second summation in \eqref{eq:main_proof_6} we use the following relation
\begin{align}
&(\bC^{(j)})^{-1}-(\bC_{[k]}^{(j)})^{-1}\notag\\
\!=&(\bC^{(j)})^{-1} \LB \bC_{[k]}^{(j)}-\bC^{(j)} \RB (\bC_{[k]}^{(j)})^{-1}\notag\\
\!=&\!-\!\frac{(\bC^{(j)})^{-1}}{M}\!\LB\! c^{(j)}_{0,k} \bh_k \bh_k^{\He}\!+\!c^{(j)}_{1,k} \bm{\delta}_k^{(j)} (\bm{\delta}_k^{(j)})^{\He}\!+\!c^{(j)}_{2,k} \bm{\delta}_k^{(j)} \bh_k^{\He}\!+\!c^{(j)}_{2,k} \bh_k (\bm{\delta}_k^{(j)})^{\He}\!\RB\!(\bC_{[k]}^{(j)})^{-1},
\label{eq:main_proof_7}
\end{align}
where $c_{0,k}^{(j)}, c_{1,k}^{(j)}, c_{2,k}^{(j)}$ is defined in \eqref{definec}. It is important to note that
\begin{align*}
c^{(j)}_{0,k} c^{(j)}_{1,k}= (c^{(j)}_{2,k})^2,c^{(j)}_{0,k}+c^{(j)}_{1,k}=1,
\end{align*}
as these relations will be used several times through the proof.

Inserting \eqref{eq:main_proof_7} into \eqref{eq:main_proof_6}, the interference term can be denoted as

\begin{align}
\tilde{\mathcal{I}}_k	
&\asymp \frac{P}{M^2} \sum_{j=1}^n\sum_{j'=1}^n \frac{\mu_j\mu_{j'}}{\sqrt{\Gamma^o_{j,j}(\bI_M)\Gamma^o_{j',j'}(\bI_M)}}\bh_k^{\He}\bE_j\bE_j^{\He}(\bC_{[k]}^{(j)})^{-1}\notag\\
&~~~~~~~~~~~~~~~~~~\cdot(\hat{\bH}_{[k]}^{(j)})^{\He}\hat{\bH}_{[k]}^{(j')}(\bC^{({j'})})^{-1} \bE_{j'}\bE_{j'}^{\He}\bh_k\notag\\
&-\frac{P}{M^3} \sum_{j=1}^n\sum_{j'=1}^n \frac{\mu_j\mu_{j'}}{\sqrt{\Gamma^o_{j,j}(\bI_M)\Gamma^o_{j',j'}(\bI_M)}}\bh_k^{\He} \bE_j\bE_j^{\He}(\bC^{(j)})^{-1}\notag\\
&~~~~~~~~~~~~~~~~~~\cdot\LSB\!\bh_k c^{(j)}_{0,k} \bh_k^{\He}\!\RSB\!(\bC_{[k]}^{(j)})^{-1}    (\hat{\bH}_{[k]}^{(j)})^{\He}\hat{\bH}_{[k]}^{(j')}(\bC^{({j'})})^{-1}\bE_{j'}\bE_{j'}^{\He} \bh_k\notag\\
&-\frac{P}{M^3} \sum_{j=1}^n\sum_{j'=1}^n \frac{\mu_j\mu_{j'}}{\sqrt{\Gamma^o_{j,j}(\bI_M)\Gamma^o_{j',j'}(\bI_M)}}\bh_k^{\He} \bE_j\bE_j^{\He}  (\bC^{(j)})^{-1}\notag\\
&~~~~~~~~~~~~~~~~~~\cdot\LSB\!\bm{\delta}_k^{(j)} c^{(j)}_{1,k}(\bm{\delta}_k^{(j)})^{\He}\!\RSB\! (\bC_{[k]}^{(j)})^{-1}     (\hat{\bH}_{[k]}^{(j)})^{\He}\hat{\bH}_{[k]}^{(j')}(\bC^{({j'})})^{-1}\bE_{j'}\bE_{j'}^{\He} \bh_k\notag\\
&-\frac{P}{M^3} \sum_{j=1}^n\sum_{j'=1}^n \frac{\mu_j\mu_{j'}}{\sqrt{\Gamma^o_{j,j}(\bI_M)\Gamma^o_{j',j'}(\bI_M)}}\bh_k^{\He} \bE_j\bE_j^{\He}  (\bC^{(j)})^{-1}\notag\\
&~~~~~~~~~~~~~~~~~~\cdot\LSB   \bm{\delta}^{(j)}_k c^{(j)}_{2,k} \bh_k^{\He}     \RSB(\bC_{[k]}^{(j)})^{-1} (\hat{\bH}_{[k]}^{(j)})^{\He}\hat{\bH}_{[k]}^{(j')}(\bC^{({j'})})^{-1}\bE_{j'}\bE_{j'}^{\He} \bh_k\notag\\
&-\frac{P}{M^3} \sum_{j=1}^n\sum_{j'=1}^n \frac{\mu_j\mu_{j'}}{\sqrt{\Gamma^o_{j,j}(\bI_M)\Gamma^o_{j',j'}(\bI_M)}}\bh_k^{\He} \bE_j\bE_j^{\He}  (\bC^{(j)})^{-1} \notag\\
&~~~~~~~~~~~~~~~~~~\cdot\LSB \!\bh_k c^{(j)}_{2,k}(\bm{\delta}_k^{(j)})^{\He}\!\RSB\! (\bC_{[k]}^{(j)})^{-1} (\hat{\bH}_{[k]}^{(j)})^{\He}\hat{\bH}_{[k]}^{(j')}(\bC^{({j'})})^{-1}\bE_{j'}\bE_{j'}^{\He} \bh_k\notag\\
&= A+B+C+D+E.
\label{eq:main_proof_8}
\end{align}

We proceed by calculating terms $A$ to $E$ in \eqref{eq:main_proof_8} successively, using Lemma~\ref{lemma2}. For the sake of simplicity, we only proceed the calculation of term $A$ and the rest terms can be calculated in similar manner.

\begin{align}
A&=\frac{P}{M^2} \sum_{j=1}^n\sum_{j'=1}^n \frac{\mu_j\mu_{j'}}{\sqrt{\Gamma^o_{j,j}(\bI_M)\Gamma^o_{j',j'}(\bI_M)}} \bh_k^{\He}\bE_j\bE_j^{\He}(\bC_{[k]}^{(j)})^{-1}\cdot(\hat{\bH}_{[k]}^{(j)})^{\He}\hat{\bH}_{[k]}^{(j')}(\bC^{({j'})})^{-1} \bE_{j'}\bE_{j'}^{\He}\bh_k\notag\\
&\asymp P\sum_{j=1}^n\sum_{j'=1}^n  \frac{\mu_j\mu_{j'}}{\sqrt{\Gamma^o_{j,j}(\bI_M)\Gamma^o_{j',j'}(\bI_M)}}\cdot\left[\frac{\trace \LB \bE_{j'}\bE_{j'}^{\He} \bm{\Theta}_k  \bE_j\bE_j^{\He}(\bC^{(j)}_{[k]})^{-1}(\hat{\bH}_{[k]}^{(j)})^{\He}\hat{\bH}_{[k]}^{(j')}(\bC_{[k]}^{({j'})})^{-1} \RB}{M^2}\right.\notag\\
&-\!c_{0,k}^{(j')}\!\frac{\trace\!\LB\!\bm{\Theta}_k\! \bE_j\!\bE_j^{\He}\!(\bC^{(j)}_{[k]})^{-1}\!(\hat{\bH}_{[k]}^{(j)})^{\He}\!\hat{\bH}_{[k]}^{(j')}\!(\bC_{[k]}^{({j'})})^{-1}\!\RB}{M^2}\cdot\frac{\trace \!\LB\!\bm{\Theta}_k\!\bE_{j'}\!\bE_{j'}^{\He}\!(\bC_{[k]}^{({j'})})^{-1}\!\RB}{M}  \frac{1\!+\!c_{1,k}^{(j')} \frac{\trace\!\LB\!\bm{\Theta}_k\!  (\bC_{[k]}^{({j'})})^{-1}\!\RB}{M}}{1\!+\!\frac{\trace\!\LB\!\bm{\Theta}_k\!  (\bC_{[k]}^{({j'})})^{-1}\!\RB}{M}}\notag\\
&\left.\!+\!(c_{2,k}^{(j')})^2\!\frac{\trace\!\LB\!\bm{\Theta}_k\! \bE_j\!\bE_j^{\He}\!(\bC^{(j)}_{[k]})^{-1}\!(\hat{\bH}_{[k]}^{(j)})^{\He}\!\hat{\bH}_{[k]}^{(j')}\!(\bC_{[k]}^{({j'})})^{-1}\!\RB}{M^2}\!\cdot\!\frac{\trace\! \LB\!\bm{\Theta}_k\!\bE_{j'}\!\bE_{j'}^{\He}\!(\bC_{[k]}^{({j'})})^{-1}\!\RB}{M}\!\frac{\frac{\trace\! \LB\!\bm{\Theta}_k\!(\bC_{[k]}^{({j'})})^{-1}\!\RB}{M}}{1\!+\!\frac{\trace\!\LB\!\bm{\Theta}_k\! (\bC_{[k]}^{({j'})})^{-1}\!\RB}{M}}\!\right].
\label{eq:main_proof_9}
\end{align}
According to Lemma \ref{lemma_rank1} and Lemma \ref{lemma1}, we can have:
\begin{align}
&\frac{1}{M^2}\trace \LB \bm{\Theta}_k\bE_j\bE_j^{\He}(\bC^{(j)}_{[k]})^{-1}(\hat{\bH}_{[k]}^{(j)})^{\He}\hat{\bH}_{[k]}^{(j')}(\bC_{[k]}^{({j'})})^{-1}\RB
\asymp \Gamma^o_{j,j'}(\bm{\Theta}_k\bE_j\bE_j^{\He})\notag\\
&\frac{1}{M^2}\!\trace\! \LB\! \bE_{j'}\bE_{j'}^{\He}\! \bm{\Theta}_k  \bE_j\bE_j^{\He}\!(\bC^{(j)}_{[k]})^{-1}\!(\hat{\bH}_{[k]}^{(j)})^{\He}\!\hat{\bH}_{[k]}^{(j')}\!(\bC_{[k]}^{({j'})})^{-1}\!\RB\!
\asymp\!\Gamma^o_{j,j'}\!(\!\bE_{j'}\bE_{j'}^{\He}\bm{\Theta}_k\bE_j\bE_j^{\He}\!).
\label{eq:main_proof_10}
\end{align}
Inserting \eqref{eq:main_proof_10} in \eqref{eq:main_proof_9} and using the fundamental Theorem~\ref{fundamental_theorem} yields
\begin{align}
A&\asymp P\sum_{j=1}^n\sum_{j'=1}^n  \mu_j\mu_{j'} \frac{\Gamma^o_{j,j'}(\bE_{j'}\bE_{j'}^{\He}\bm{\Theta}_k\bE_j\bE_j^{\He})}{\sqrt{\Gamma^o_{j,j}(\bI_M)\Gamma^o_{j',j'}(\bI_M)}}\notag\\
&~~~~~~~~~~~~~-\mu_j\mu_{j'}c_{0,k}^{(j')}\frac{\Gamma^o_{j,j'}(\bm{\Theta}_k\bE_j\bE_j^{\He})}{\sqrt{\Gamma^o_{j,j}(\bI_M)\Gamma^o_{j',j'}(\bI_M)}}   \frac{\trace\LB\bm{\Theta}_k\bE_{j'}\bE_{j'}^{\He}\bQ_o^{(j')}\RB}{M}   \frac{1+c_{1,k}^{(j')}m^{(j')}_{k}}{1+m^{(j')}_{k}}\notag\\
&~~~~~~~~~~~~~+\mu_j\mu_{j'}(c_{2,k}^{(j')})^2\frac{\Gamma^o_{j,j'}(\bm{\Theta}_k\bE_j\bE_j^{\He})}{\sqrt{\Gamma^o_{j,j}(\bI_M)\Gamma^o_{j',j'}(\bI_M)}}  \frac{\trace\LB\bm{\Theta}_k\bE_{j'}\bE_{j'}^{\He}\bQ_o^{(j')}\RB}{M} \frac{m^{(j')}_{k}}{ 1+m^{(j')}_{k}}\notag\\
&\stackrel{(a)}{\asymp} P\sum_{j=1}^n\sum_{j'=1}^n  \mu_j\mu_{j'}\frac{\Gamma^o_{j,j'}(\bE_{j'}\bE_{j'}^{\He}\bm{\Theta}_k\bE_j\bE_j^{\He})}{\sqrt{\Gamma^o_{j,j}(\bI_M)\Gamma^o_{j',j'}(\bI_M)}}\notag\\
&~~~~~~~~~~~~~-\mu_j\mu_{j'}c_{0,k}^{(j')}\frac{\Gamma^o_{j,j'}(\bm{\Theta}_k\bE_j\bE_j^{\He})}{\sqrt{\Gamma^o_{j,j}(\bI_M)\Gamma^o_{j',j'}(\bI_M)}}   \frac{\trace\LB\bm{\Theta}_k\bE_{j'}\bE_{j'}^{\He}\bQ_o^{(j')}\RB}{M}\frac{1}{1+m^{(j')}_{k}},\notag
\end{align}
where equality $(a)$ follows from $c^{(j')}_{0,k} c^{(j')}_{1,k}= (c^{(j')}_{2,k})^2$.
Proceed similarly for the remaining~$4$ terms and add term A, B, C, D and E together, we can get

\begin{align}
\mathcal{I}_k&= A+B+C+D+E \notag\\
&=P\sum_{j=1}^n\sum_{j'=1}^n  \frac{\mu_j\mu_{j'}\Gamma^o_{j,j'}(\bE_{j'}\bE_{j'}^{\He}\bm{\Theta}_k\bE_j\bE_j^{\He})}{\sqrt{\Gamma^o_{j,j}(\bI_M)\Gamma^o_{j',j'}(\bI_M)}}\notag\\
&-2P\sum_{j=1}^n\sum_{j'=1}^n  \frac{\mu_j\mu_{j'}\Gamma^o_{j,j'}(\bm{\Theta}_k\bE_j\bE_j^{\He})}{\sqrt{\Gamma^o_{j,j}(\bI_M)\Gamma^o_{j',j'}(\bI_M)}}   \frac{\trace\LB\bm{\Theta}_k\bE_{j'}\bE_{j'}^{\He}\bQ_o^{(j')}\RB}{M}\frac{c_{0,k}^{(j')}}{1+m^{(j')}_{k}}\notag\\
&+P\sum_{j=1}^n\sum_{j'=1}^n \mu_j\mu_{j'} \frac{\frac{\trace\LB\bm{\Theta}_k\bE_{j'}\bE_{j'}^{\He}\bQ_o^{(j')}\RB}{M}\frac{\trace\LB\bm{\Theta}_k\bE_{j}\bE_{j}^{\He}\bQ_o^{(j)}\RB}{M}\Gamma^o_{j,j'}(\bm\Theta_k)}{\sqrt{\Gamma^o_{j,j}(\bI_M)\Gamma^o_{j',j'}(\bI_M)}}\notag\\
&~~~~~~~~~~~~~~\cdot\frac{c_{0,k}^{(j)}c_{0,k}^{(j')}+\rho_k^{(j,j')}c_{2,k}^{(j)}c_{2,k}^{(j')}}{(1+m_k^{(j)})(1+m_k^{(j')})}.\notag
\end{align}
This concludes the proof.
\FloatBarrier

\section{Proof of Power Allocation Theorem~\ref{iterationmethod}}\label{se:proof2}

Since $\bm\mu^{[t+1]}$ minimize the optimization problem in step $5$ in Algorithm \ref{iterativeCCP}, we can have
\begin{align*}
\sum_{k=1}^K\lambda_k^{[t]}\cdot \frac{\frac{1}{P}+(\bm\mu^{[t+1]})^{\Te}\bB_k\bm\mu^{[t+1]}}{\frac{1}{P}+(\bm\mu^{[t+1]})^{\Te}\LB\bA_k+\bB_k\RB\bm\mu^{[t+1]}}\leq \sum_{k=1}^K\lambda_k^{[t]}\cdot \frac{\frac{1}{P}+(\bm\mu^{[t]})^{\Te}\bB_k\bm\mu^{[t]}}{\frac{1}{P}+(\bm\mu^{[t]})^{\Te}\LB\bA_k+\bB_k\RB\bm\mu^{[t]}},
\end{align*}
Insert the expression for $\lambda_k^{[t]}$ in \eqref{expression_lambda}, use the notation for $u_k$ defined in \eqref{defu}, the above expression simplifies as
\begin{align*}
\sum_{k=1}^K\frac{u_k(\bm\mu^{[t+1]})}{u_k(\bm\mu^{[t]})}\leq K.
\end{align*}
According to AM-GM inequality
\begin{align*}
K\sqrt[K]{\frac{\prod_{k=1}^Ku_k(\bm\mu^{[t+1]})}{\prod_{k=1}^Ku_k(\bm\mu^{[t]})}}\leq\sum_{k=1}^K\frac{u_k(\bm\mu^{[t+1]})}{u_k(\bm\mu^{[t]})},
\end{align*}
we can obtain
\begin{align*}
\prod_{k=1}^Ku_k(\bm\mu^{[t+1]})\leq\prod_{k=1}^Ku_k(\bm\mu^{[t]}).
\end{align*}
This shows that the value $\prod_{k=1}^Ku_k(\bm\mu)$ decreases during the iteration for updating $\bm\mu$. Since the physical meaning for $\prod_{k=1}^Ku_k(\bm\mu)$ is the sequence product of the MSE at each RX and therefore $\prod_{k=1}^Ku_k(\bm\mu)>0$. According to monotone convergence theorem, the iterative algorithm will produce a decreasing and lower bound series of MSE sequence product while updating $\bm\mu$, therefore the iterative procedure is surely to converge to a local optimum. This completes the proof.


\bibliographystyle{IEEEtran}
\bibliography{Literature}
\end{document}

%% file: Definitions.tex
\newtheorem{theorem}{Theorem}

\newtheorem{lemma}{Lemma}

\newtheorem{example}{Example}

\DeclareMathAlphabet{\mathbit}{OML}{cmr}{bx}{it}
\DeclareMathAlphabet{\mathsf}{OT1}{cmss}{m}{n}
\DeclareMathAlphabet{\mathTXf}{OT1}{cmss}{bx}{it}

\DeclareMathOperator{\diag}{diag}

\DeclareMathOperator{\trace}{tr}

\DeclareMathOperator*{\argmin}{argmin}
\DeclareMathOperator*{\argmax}{argmax}

\DeclareMathOperator{\rZF}{rZF}

\DeclareMathOperator{\TX}{TX}

\DeclareMathOperator{\DCSI}{DCSI}
\DeclareMathOperator{\Rate}{R}
\DeclareMathOperator{\CN}{\mathcal{N}_{\mathbb{C}}}

\DeclareMathOperator{\SINR}{SINR}

\newcommand{\bA}{\mathbf{A}}
\newcommand{\bB}{\mathbf{B}}
\newcommand{\bC}{\mathbf{C}}

\newcommand{\bE}{\mathbf{E}}

\newcommand{\bH}{\mathbf{H}}
\newcommand{\bI}{\mathbf{I}}

\newcommand{\bL}{\mathbf{L}}

\newcommand{\bQ}{\mathbf{Q}}
\newcommand{\bR}{\mathbf{R}}

\newcommand{\bT}{\mathbf{T}}
\newcommand{\bU}{\mathbf{U}}
\newcommand{\bV}{\mathbf{V}}

\newcommand{\bX}{\mathbf{X}}

\newcommand{\bb}{\bm{b}}

\newcommand{\be}{\bm{e}}

\newcommand{\bh}{\bm{h}}

\newcommand{\bn}{\bm{n}}

\newcommand{\bq}{\bm{q}}

\newcommand{\bs}{\bm{s}}
\newcommand{\bt}{\bm{t}}

\newcommand{\bz}{\bm{z}}

\newcommand{\LB}{\left(}
\newcommand{\RB}{\right)}
\newcommand{\LSB}{\left[}
\newcommand{\RSB}{\right]}

\newcommand{\I}{\mathbf{I}} 
\newcommand{\norm}[1]{\lVert{#1}\rVert}

\newcommand{\Fro}{{\mathrm{F}}}
\newcommand{\tr}{{\text{tr}}}

\newcommand{\E}{{\mathbb{E}}}

\newcommand{\He}{{{\mathrm{H}}}}
\newcommand{\Te}{{{\mathrm{T}}}}

\newcommand{\xv}{\mathbf{x}}
\newcommand{\yv}{\mathbf{y}}
\newcommand{\zv}{\mathbf{z}}
\newcommand{\qv}{\mathbf{q}}
\newcommand{\wv}{\mathbf{w}}

\theoremstyle{remark}
\newtheorem{remark}{Remark} 
\theoremstyle{assumption}
\newtheorem{assumption}{Assumption} 